\documentclass[journal,twoside,web]{ieeecolor}
\usepackage{etoolbox}
\makeatletter
\@ifundefined{color@begingroup}%
{\let\color@begingroup\relax
	\let\color@endgroup\relax}{}%
\def\fix@ieeecolor@hbox#1{%
	\hbox{\color@begingroup#1\color@endgroup}}
\patchcmd\@makecaption{\hbox}{\fix@ieeecolor@hbox}{}{\FAILED}
\patchcmd\@makecaption{\hbox}{\fix@ieeecolor@hbox}{}{\FAILED}
\usepackage{generic}
\usepackage{cite}
\usepackage{amsmath,amssymb,amsfonts,arydshln}
\usepackage{stfloats}
\usepackage{cuted} 
\usepackage{algorithm,algorithmic}
\newcommand{\removelatexerror}{\let\@latex@error\@gobble}
\makeatother
\usepackage{caption} 
\usepackage{lipsum} 
\usepackage{units}
\usepackage[dvipsnames]{xcolor}
\definecolor{DR}{RGB}{192, 0, 0}
\definecolor{LR}{RGB}{232, 86, 66}
\definecolor{LB}{RGB}{18, 128, 176}
\definecolor{DB}{RGB}{37, 83, 125}
\definecolor{LG}{RGB}{20,187,99}
\definecolor{DG}{RGB}{32,108,23}
\usepackage{graphicx}
\usepackage{textcomp}
\usepackage{booktabs}
\usepackage{hyperref}
\hypersetup{hidelinks=true}
\newtheorem{theorem}{Theorem}
\newtheorem{lemma}{Lemma}
\newtheorem{remark}{Remark}

\newtheorem{definition}{Definition}
\newtheorem{example}{Example}
\newtheorem{problem}{Problem}

\def\BibTeX{{\rm B\kern-.05em{\sc i\kern-.025em b}\kern-.08em
    T\kern-.1667em\lower.7ex\hbox{E}\kern-.125emX}}

\begin{document}
\title{Open/Closed-loop Active Learning for Data-driven Predictive Control}
\author{Shilun Feng, Dawei Shi, Yang Shi, and Kaikai Zheng
\thanks{Shilun Feng, Dawei Shi, and Kaikai Zheng are with the State Key Laboratory of Intelligent Control and Decision of Complex Systems and MIIT Key Laboratory of Servo Motion System Drive and Control, School of Automation, Beijing Institute of Technology, Beijing 100081, China (e-mail: shilunfeng@bit.edu.cn; daweishi@bit.edu.cn; kaikai.zheng@bit.edu.cn).}
\thanks{
Yang Shi is with the Department of Mechanical Engineering, Faculty of Engineering, University of Victoria, Victoria, BC V8N 3P6, Canada (e-mail: yshi@uvic.ca).}
}

\maketitle

\begin{abstract}
    An important question in data-driven control is how to obtain an informative dataset. In this work, we consider the problem of effective data acquisition of an unknown linear system with bounded disturbance for both open-loop and closed-loop stages. The learning objective is to minimize the volume of the set of admissible systems. First, a performance measure based on historical data and the input sequence is introduced to characterize the upper bound of the volume of the set of admissible systems. On the basis of this performance measure, an open-loop active learning strategy is proposed to minimize the volume by actively designing inputs during the open-loop stage. For the closed-loop stage, a closed-loop active learning strategy is designed to select and learn from informative closed-loop data. The efficiency of the proposed closed-loop active learning strategy is proved by showing that the unselected data cannot benefit the learning performance. Furthermore, an adaptive predictive controller is designed in accordance with the proposed data acquisition approach. The recursive feasibility and the stability of the controller are proved by analyzing the effect of the closed-loop active learning strategy. Finally, numerical examples and comparisons illustrate the effectiveness of the proposed data acquisition strategy.
\end{abstract}


\section{Introduction}
\label{sec:introduction}
Due to the advancements in sensing technology and cyber-physical systems, there has been a growing interest in designing controllers directly from data \cite{hou2013from_mbc2ddc, depersis2020formulas, baggio2021NC_DDC}. As a process of learning from data, a major focus in data-driven control has been on how to enhance the learning efficiency. This becomes imperative when resources are limited, or when the control mission needs to be accomplished within a limited time window. Moreover, data-driven control for linear time-invariant (LTI) systems has been widely investigated since various systems behave as LTI systems around their operating points in practice \cite{johanssonQuadrupletankProcessMultivariable2000,spongNonlinearControlReaction2001a}. Thus, {a central challenge} in this context is the development of an effective solution that actively generates and selects informative data samples for unknown LTI systems. Such a solution would promisingly lead to a satisfactory controller design with a smaller size of data samples, reduced data transmission, and shortened learning time.

{When data is randomly generated and collected, learning is often slow and inefficient due to the presence of redundant information \cite{LANG1995331}. {Active learning offers a promising solution to this problem due to its capability to selectively acquire informative data for learning \cite{settles.tr09}.} {Early methods for data generation were proposed by either estimating the potential response of the system using a preliminary model \cite{geversInputDesignOpenloop2006},} or by evaluating the performance of the current controller \cite{hjalmarssonExperimentDesignClosedloop2005}. From a systems and control perspective, active learning can be divided into open-loop active learning and closed-loop active learning, corresponding to open-loop and closed-loop control systems, respectively. For both cases, the objective is to maximize a certain data informativity metric. For the purpose of data collection, open-loop active learning approaches design inputs to drive the system to the desirable state. In contrast, closed-loop active learning maximizes the metric by actively selecting closed-loop data, since the system is driven by the closed-loop controller.} Generally, active learning can be classified by the employed metrics. For {open-loop active learning}, metrics are designed to quantify uncertainty \cite{Lewis1994uncertaintysampling, pmlr-v162-raj22a}, diversity \cite{10.1145/1015330.1015349}, or various information measures \cite{Yu2006Active,PRONZATO2008303,Abraham2019activelearning,2022_Beyond,Monica2022online,muller2023inputdesign}, e.g., Fisher information matrix and Hankel matrix. We refer readers to \cite{pmlr-v125-wagenmaker20a,taylor2021ALinRobotics} for additional designs of metrics and their applications in active learning. Open-loop active learning strategies need to actively excite the system to collect desirable data samples. However, in closed-loop systems, the input cannot be arbitrarily designed by the user to generate desirable samples, since the inputs are generated by a closed-loop controller.

Closed-loop data can also be leveraged to learn the system dynamics online, which is consistent with the considerations in closed-loop system identification \cite{vandenhofClosedloopIssuesSystem1998}. For data acquisition in closed-loop scenarios, an important issue is how to make better use of the data generated by the feedback controller. A natural idea is to collect each data sample. However, a main drawback of this approach is that the online learning algorithm need to be consecutively performed with the updated dataset. In response to this conundrum, closed-loop active learning schemes received considerable attention in recent years for their ability to actively select informative closed-loop data based on a certain criterion \cite{Heemels2012IntroETC,peng2018survey}. Some of the early research contributes to state estimation, where only informative data samples are collected to estimate the true state of the system \cite{Shi2014event,Shi2015setvalue}. Inspired by this idea, there has been an increasing interest in learning unknown system dynamics by actively selecting closed-loop data. {Several studies utilized the performance of the learning result as the learning criterion \cite{Umlauft2020ETL_Gaussian,guo2020prediction,Solowjow2018ETL,Solowjow2020ETL,umlauft2020smart,lederer2021gaussian,jiao2022backstepping,he2022learning,Beuchert2020ETL_Cyclically,zheng2022economic}. For instance, Umlauft and Hirche considered Gaussian processes and utilized the uncertainty of the learned model as the criterion \cite{Umlauft2020ETL_Gaussian}. Solowjow \emph{et al.} considered linear systems and set the mismatch of communication rates between the learned model and the true model as the criterion \cite{Solowjow2018ETL}. In addition to the learning results, some other works focused on the properties of the {closed-loop data} samples \cite{Zheng2023ETL_LACKI,DIAO2018ETL_FIR,gatsis2021adaptive,gatsis2022federated}. As an example, Zheng \emph{et al.} employed the distance between the new data sample and the dataset as the learning criterion for lazily adapted constant kinky inference problems \cite{Zheng2023ETL_LACKI}. For FIR models, Diao \emph{et al.} used the magnitude of data variation as the criterion \cite{DIAO2018ETL_FIR}. However, analyzing the informativity of closed-loop data samples of linear systems without prior knowledge of the true system remains an open challenge.}

{Another topic related to our work is adaptive predictive control, which is widely used due to its ability to handle parameter uncertainties and time-variant parameters \cite{Qin2003Survey_MPC}. These advantages allow us to address common issues in the {data informativity framework \cite{2022_SLemma}}, such as the update of the dataset and the learning uncertainty introduced by noisy datasets.} The research on model-based adaptive predictive control can be traced back to \cite{Mayne1993AMPC}, where nonlinear systems with input constraints were analyzed. Later, several works discussed the convergence of parameter estimates by introducing data-selection methods \cite{Fukushima2007AMPC} or additional constraints \cite{vicente2019stabilizing}. Another line of research integrated set-membership identification with robust predictive control where various representations of uncertainties were constructed \cite{Adetola2011AMPC_Set,Lu2019AMPC_Set,Lorenz2019RMPC_update, Zhang2020AMPC_uncertainty}. To enable data-driven predictive control for LTI systems, {the approaches based on the fundamental lemma} and the set-valued approaches are widely investigated. The fundamental lemma, which was proposed by Willems \textit{et al.} in \cite{willemsNotePersistencyExcitation2005}, illustrates that {all trajectories of a deterministic LTI system can be captured by a finite collection of noise-free trajectories. To enable the application of this lemma in the predictive control of noisy systems, slack variables are employed to guarantee the feasibility of the predictor \cite{coulson2019deepc}.} With this lemma, various data-driven formulations were proposed for different predictive control configurations \cite{mehrnooshEventTriggeredRobustDataDriven2023,liuDataDrivenSelfTriggeringMechanism2023,liuDataDrivenResilientPredictive2023, schmittDataDrivenPredictiveControl2024, barosOnlineDataenabledPredictive2022}. On the other hand, the set-valued approaches utilize a data-driven formulation of admissible systems consistent with the collected data samples. {Some studies analyzed the reachable system states under a given input sequence based on the set of admissible systems and designed robust predictive controllers using these reachable states to replace the prediction of system trajectories given by the true model \cite{russoTubeBasedZonotopicDataDriven2023,alanwarRobustDatadrivenPredictive2022}. {However, enabling a stabilizing controller with a reduced number of data samples by actively selecting informative ones remains unclear for data-driven predictive control.} This gap also motivates our investigations in this work. Note that our work is also different from existing results on data-driven predictive control with event-/self-triggered and denial communications \cite{mehrnooshEventTriggeredRobustDataDriven2023,liuDataDrivenSelfTriggeringMechanism2023,liuDataDrivenResilientPredictive2023}, where the focuses were on exploring event-triggered/denial data transmission and controller update mechanisms to mitigate communication constraints and ensure stability and performance of the closed-loop system. By contrast, our work considers the scenario where the data generation and transmission processes are not restricted, and our aim is to design learning methods to actively generate and select informative data samples for enhanced data efficiency and accelerated learning by evaluating the informativity of the data samples.}

{In this work, we consider the scenario that the learning framework consists of two stages: an open-loop stage and a closed-loop stage. During the open-loop stage, the data acquisition is driven by an active learning strategy to generate informative data samples. In the closed-loop stage, a data-driven adaptive predictive controller with an active learning strategy is proposed. The proposed closed-loop active learning strategy can detect and learn from informative closed-loop data samples, thereby enhancing control performance.} A few challenges, however, need to be addressed to fulfill the design objectives. First, the informativity measure of datasets depends on the unknown system dynamics, which makes it difficult to suitably design the input signals during the {open-loop stage}. Second, the determination of the informativity of a single data sample could be computationally expensive. {The informativity of a single data sample is related not only to the learning result based on the previous dataset but also to the learning result based on the updated dataset that includes the incoming data sample. Therefore, a computationally efficient metric to evaluate the informativity of a data sample need to be proposed.} Furthermore, the effect of learning-based adaptation also adds to the difficulty of analyzing the closed-loop stability and the recursive feasibility of the predictive controller. The main contributions of this work are summarized as follows:
\begin{enumerate}
    \item {For linear systems with bounded process noises, a necessary and sufficient condition is proposed to ensure the contraction of the upper bound of the volume of the  admissible system ellipsoids (Theorem 1). An open-loop active learning strategy is then designed to minimize the volume by actively designing system inputs based on historical state-input trajectories.}
    \item To determine the informativity of a single data sample without actually running the learning algorithm, {an informativity criterion is formulated based on a simple semidefinite program.} {Then, a {closed-loop active learning} strategy which actively selects and learns from the closed-loop data is proposed based on this criterion.} We show the inclusion of unselected data samples cannot benefit the learning algorithm to obtain a smaller set of admissible systems and prove the feasibility of the learning algorithm (Theorem \ref{thm:ETL}).
    \item In the {closed-loop stage}, an adaptive tube-based data-driven predictive controller is designed by applying the proposed closed-loop active learning strategy. We first show the contraction of tubes under the learning and adaptation strategy (Lemma \ref{lem:error_set_inclusion}). Furthermore, given a suitable dataset, the recursive feasibility and the stability of the proposed learning-based adaptive predictive controller are proved (Theorem \ref{thm:ATDPC_stable}).
\end{enumerate}

The rest of the paper is organized as follows. Section \ref{sec:Preliminaries} provides the preliminaries of this work. {Section \ref{sec:Problem Formulation} presents the problems considered in this work.} In Section \ref{sec:MoASE}, the approximation of the volume of admissible systems sets is presented. In Section \ref{sec:Active Online Learning for Modeling}, the active learning strategy is presented. Then, an adaptive predictive controller with a closed-loop active learning strategy is proposed, and sufficient conditions for recursive feasibility and stability are established in Section \ref{sec:Tube-based Data-driven Predictive Control}. Several numerical examples are presented and discussed in Section \ref{sec:Numerical Examples}.

\emph{Notation: }For a matrix $S$, $S_{[i,j]}$ denotes the element at the $i$th row and $j$th column of $S$, $S_{[i]}$ denotes the $i$th column of it, $S_{[i:j]}$ denotes the matrix $\![ S_{[i]} \ S_{[i+1]} \dots \ S_{[j]} ]$, $\!S^{\top}\!$ denotes its transpose, $\!\det(S)\!$ denotes its determinant, $\!\text{col}(S)\!$ denotes the number of its columns, and $\text{vec}(S)\!:=\![ S_{[1]}^{\top} \ S_{[2]}^{\top}  \dots \ S_{[\text{col}(S)]}^{\top} ]^{\!\top}\!\!$ denotes its vectorization. We write $\!S\!\succ\!0\!$ ($S\!\succeq\!0$) if $\!S\!$ is symmetric and positive (semi-)definite, and we denote negative (semi-)definiteness similarly. Besides, we let $\mathbb{Z}_{\geq0}$ denote the set of nonnegative integers and write $\mathbb{Z}_{[i,j]}:=\{i,i+1,\dots,j\}$. Let $\Vert \cdot \Vert_2$ denote the 2-norm of vectors or matrices. {We denote by $I_n$ the identity matrix of dimension $n$, and by $0_{n}$ the square matrix of dimension $n$, where all elements are zero. When subscripts are not specified, $I$ and $0$ represent the identity matrix and zero matrix of appropriate dimensions, respectively.} For two sets $\mathbb{X}$ and $\mathbb{Y}$, Minkowski set addition is defined as $\mathbb{X} \oplus \mathbb{Y}:=\{x+y \ | \ x\in\mathbb{X}, y\in \mathbb{Y}\} $, and Minkowski set subtraction is defined similarly as $\mathbb{X} \ominus \mathbb{Y}:=\{x-y \ | \ x\in\mathbb{X}, y\in \mathbb{Y}\}$. {We also define $\mathcal{D}_n:=\{\text{diag}(d_1,d_2,\dots,d_n) \ | \ d_i \geq 0, i \in \mathbb{Z}_{[1,n]}\}$.}

\section{Preliminaries} \label{sec:Preliminaries}

In this paper, we consider an LTI system
\begin{equation} \label{eq:sys_real}
    x(k+1) = \mathcal{A}x(k)+\mathcal{B}u(k)+w(k)
\end{equation}
where $x(k) \!\in\! \mathbb{R}^{n_x}$ is the state, $u(k)\!\in\! \mathbb{R}^{n_u}$ is the input, $w(k)\!\in \!\mathbb{R}^{n_x}$ is the unknown disturbance, and $\mathcal{A}\!\in\! \mathbb{R}^{n_x \times n_x}$ and $\mathcal{B}\!\in\! \mathbb{R}^{n_x \times n_u}$ denote the unknown state and input matrices. We assume that $w(k) \in \mathcal{W}$ holds for all $k\in \mathbb{Z}_{\geq 0}$, where $\mathcal{W} := \left\{ w \in \mathbb{R}^{n_x} \mid w^{\top}w < \delta\right\}$, {with $\delta>0$ being a known scalar}. For compactness, we define {
\begin{equation}
	\setlength\belowdisplayskip{-6pt}
	\Delta_r:=[\mathcal{A} \ \mathcal{B}], \ h(k):=[x(k)^{\!\top} u(k)^{\!\top}]^{\top}\!, \ n_h := n_x+n_u.
\end{equation}}
\subsection{Matrix Ellipsoids and the Matrix S-Lemma}
{
	Before formulating the problem of active learning for adaptive data-driven predictive control, preliminaries about matrix ellipsoids and the matrix S-lemma are introduced first for the legibility and clarity of this work.
	\begin{definition}[\hspace{-.01em}\cite{2021_Trade-off}] \label{def:matrix_ellipsoid}
		For symmetric matrices $E \!\in\! \mathbb{R}^{p \times p}$, $G \!\in\! \mathbb{R}^{q \times q}$ and matrix $F \!\in\! \mathbb{R}^{p \times q}$, a \emph{matrix ellipsoid} is defined as a set in the following form:
		\begin{equation} \label{eq:matrix_ellipsoid_1}
			\begin{split}
				{\mathcal{E}_{m}} \!:=\! \left\{ {Z \! \in \mathbb{R}^{p \times q} \hspace{2pt}|\hspace{2pt} Z^{\top}EZ+Z^{\top}F+F^{\top}Z+G \preceq 0 }\right\}
			\end{split}
		\end{equation}
		where $E \!\succ\! 0$, ${F^{\top}E^{-1}F-G \!\succ\! 0}$. If we let ${Z_c\!:=\!-E^{-1}F}$, $G_c\!:=\!F^{\top}E^{-1}F-G$, the set in \eqref{eq:matrix_ellipsoid_1} can be reformulated as:
		\begin{equation} \label{eq:matrix_ellipsoid_2}
			{\mathcal{E}^{'}_{m}} := \left\{ {Z \in \mathbb{R}^{p \times q} \hspace{2pt}|\hspace{2pt} (Z-Z_c)^{\top}E(Z-Z_c) \preceq G_c }\right\},
			\end{equation}
		which is another form of matrix ellipsoid. ${Z_c}$ is the \emph{center} of $\mathcal{E}_{m}$. Besides, we also note that ${Z_c}$ always belongs to $\mathcal{E}_{m}$ on the basis of \eqref{eq:matrix_ellipsoid_2}.
	\end{definition}

	As an extension to the classical ellipsoids in the Euclidean space \cite{1994_LMI}, the constraint $E \succ 0$ ensures that $\mathcal{E}_{m}$ is bounded, and the constraint $F^{\top}E^{-1}F - G \succ 0$ ensures that $\mathcal{E}_{m}$ is not empty or a singleton. Furthermore, the measure of matrix ellipsoids was analyzed  using standard measure theory in \cite{2021_Trade-off}.
	\begin{lemma}[\text{see \cite[Lemma 1]{2021_Trade-off}}] \label{Lem:MoME}
		A matrix ellipsoid is measurable. For $\mathcal{E}_m$ defined in \eqref{eq:matrix_ellipsoid_1}, a measure can be defined as $\mu(\mathcal{E}_{m})\!:=\!m(\mathcal{V}(\mathcal{E}_{m}))$, where $\!\mathcal{V}(\mathcal{E}_{m})\!:=\!\{\text{vec}(s) \mid s\!\in\!\mathcal{E}_{m}\}\!$ and $m(\mathcal{V}(\mathcal{E}_{m}))$ is the Lebesgue measure of $\mathcal{V}(\mathcal{E}_{m})$. Accordingly, $\mu(\mathcal{E}_{m})=\beta_m (\det{( F^{\top}E^{-1}F - G)})^{p/2}(\det{(E^{-1})})^{q/2}$, where the constant $\beta_m$ depends only on \( p \) and \( q \), which represent the number of rows and columns, respectively, of the matrices belonging to \(\mathcal{E}_{m}\).
	\end{lemma}

	Note that any element in a matrix ellipsoid satisfies a quadratic matrix inequality (QMI) according to \eqref{eq:matrix_ellipsoid_1}. To offer an insight into QMIs, a generalized version of the S-lemma was discussed in \cite{2022_SLemma}.
	\begin{lemma}[\text{\hspace{-.01em}\cite[Theorem 9]{2022_SLemma}}] \label{lem:S-lemma}
		Let $\tilde{Z} \in \mathbb{R}^{n_a \times n_b}$, $\tilde{M}$ and $\tilde{N}$ be symmetric matrices in $\mathbb{R}^{(n_a+n_b) \times (n_a+n_b)}$. Assume that there exists some matrix $\bar{Z} \in \mathbb{R}^{n_a \times n_b}$ such that
		\begin{equation}  \label{eq:SLemma_Slater}
			\begin{bmatrix}
				I \\ \bar{Z}
			\end{bmatrix}^{\!\top}\!
			\tilde{N}
			\begin{bmatrix}
				I \\ \bar{Z}
			\end{bmatrix}
			\succ 0.
		\end{equation}
		Then, the following statements are equivalent:
		\begin{enumerate}
			\item
			$\left[\begin{smallmatrix}
				I \\ \tilde{Z}
			\end{smallmatrix}\right]^{\!\top} \hspace{-0.3em} \tilde{M}
			\left[\begin{smallmatrix}
				I \\ \tilde{Z}
			\end{smallmatrix}\right] \!\succeq\! 0, \forall \tilde{Z} \in \mathbb{R}^{n_a \times n_b} \hspace{0.2em} \text{with}
			\left[\begin{smallmatrix}
				I \\ \tilde{Z}
			\end{smallmatrix}\right]^{\!\top} \hspace{-0.3em} \tilde{N}
			\left[\begin{smallmatrix}
				I \\ \tilde{Z}
			\end{smallmatrix}\right] \!\succeq\! 0.$ \\[-.8em]
			\item There exists a scalar $\alpha \geq 0$ such that $\tilde{M}-\alpha \tilde{N} \succeq 0.$
		\end{enumerate}
	\end{lemma}
	
	{\begin{remark} \label{rem:onedirect_slemma}
	 	Note that in Lemma \ref{lem:S-lemma}, the first statement follows directly from the second statement without the verification of the generalized Slater condition. Thus, we do not need to check the generalized Slater condition \eqref{eq:SLemma_Slater} to prove the implication ``2)'' $\Rightarrow$ ``1)'' in Lemma \ref{lem:S-lemma}.
	\end{remark}}
}
\subsection{Admissible System Ellipsoids}
{Suppose that the collected data are $u_s(i)$, $i \in \mathbb{Z}_{[0,t_0-1]}$, and $x_s(i)$, $i \!\in\! \mathbb{Z}_{[0,t_0]}$, where $u_s(i)$ denotes the applied input sequence, and $x_s(i)$ denotes the corresponding state trajectory.} Then, following the general framework for data-driven analysis and control in \cite{2022_SLemma}, we term a matrix pair $(A,B)$ an \emph{admissible system} if there exist $w_s(i) \! \in \mathcal{W}$, $i\in \mathbb{Z}_{[0,t_0-1]}$ such that
\begin{equation} \label{eq:sys_cal}
	x_s(i+1) = Ax_s(i)+Bu_s(i)+w_s(i), \forall i\in \mathbb{Z}_{[0,t_0-1]}.
\end{equation}
We define $\!\Delta\!:=\! [A \ B]$, {$h_s(i)\!:=\![x_s(i)^{\!\top} u_s(i)^{\!\top}]^{\top}\!\!$, and
\begin{equation*}
	\xi(h,\dot{x}):=\begin{bmatrix}
		I &\hspace{-2pt} \dot{x} \\
		0 &\hspace{-2pt} - h
	\end{bmatrix}
	\begin{bmatrix}
		\delta I &\hspace{-2pt} 0 \\
		0 &\hspace{-2pt} -I
	\end{bmatrix}
	\begin{bmatrix}
		I &\hspace{-2pt} \dot{x} \\
		0 &\hspace{-2pt} - h
	\end{bmatrix}^{\!\top}\!\!.
\end{equation*}
}{Note that the notation ``dot'' denotes forward time advance in discrete-time systems following the convention in sampled-data systems \cite{chen2012optimal}.}

Then, for any admissible system $(A,B)$, we have
\begin{equation} \label{eq:ase_single}
	\begin{bmatrix}
		I \\[2pt] \Delta^{\top}
	\end{bmatrix}^{\top}\!
	\xi(h_s(i),x_s(i+1))
	\begin{bmatrix}
		I \\[2pt] \Delta^{\top}
	\end{bmatrix}
	\succ 0
\end{equation}
holds for all $i\in\mathbb{Z}_{[0,t_0-1]}$ based on \eqref{eq:sys_cal} and the definition of admissible systems. We denote the number of collected data samples as $n$ and group these samples as
{\begin{equation*}
	\begin{split}
		H &:= [h_s({d_1}), h_s({d_2}),\dots, h_s(d_n)], \\
		\dot{X} &:= [x_s({d_1\!+\!1}), x_s({d_2\!+\!1}),\dots, x_s({d_n\!+\!1})], \\
	\end{split}
\end{equation*}
}where $\{d_1, d_2, \dots, d_n\} \!\subseteq\! \mathbb{Z}_{[0,t_0-1]}$ is {the set of indices} of these samples. Suppose that $W:=[w_s({d_1}), w_s({d_2}),\dots, w_s(d_n)]$. Then, according to \eqref{eq:sys_real}, we have
\begin{equation} \label{eq:matrix_sys}
	\dot{X}=\Delta_r H+W.
\end{equation}
Furthermore, for {$\Lambda\!=\!\text{diag}(\lambda_1,\lambda_2,\dots, \lambda_n)\!\in\!\mathcal{D}_n$}, we define
\begin{equation*}
	{\Xi(H,\dot{X},\Lambda)\!:=\!
	\begin{bmatrix}
		I & \dot{X} \\
		0 & -H
	\end{bmatrix}\!
	\begin{bmatrix}
		\mathrm{tr}(\Lambda) \delta I &\hspace{-4pt} 0 \\
		0 &\hspace{-4pt} -\Lambda
	\end{bmatrix}\!
	\begin{bmatrix}
		I & \dot{X} \\
		0 & -H
	\end{bmatrix}^{\!\top}\!.}
\end{equation*}
We can also derive $\Xi(H,\dot{X},\Lambda) = \sum_{i=1}^{n}\lambda_i \xi(H_{[i]},\dot{X}_{[i]})$. Then, on the basis of \eqref{eq:ase_single}, it follows directly that
\begin{equation} \label{eq:admissible_system}
	\begin{bmatrix}
		I \\ \Delta^{\top}
	\end{bmatrix}^{\top}\!
	\Xi(H,\dot{X},\Lambda)
	\begin{bmatrix}
		I \\ \Delta^{\top}
	\end{bmatrix}
	\succeq 0
\end{equation}
holds for any admissible system $(A,B)$. Now, we can formulate a set that captures all admissible systems as
\begin{equation} \label{eq:ASE}
    \mathcal{C}(H,\dot{X},\Lambda) := \left\{(A,B) \text{ $\vert$ \eqref{eq:admissible_system} is satisfied} \right\}.
\end{equation}
{Note that $\lambda_i$ acts as the weight of $\xi(H_{[i]},\dot{X}_{[i]})$ that {constitutes} $\Xi(H,\dot{X},\Lambda)$, thereby controlling the shape of $\mathcal{C}(H,\dot{X},\Lambda)$.} In fact, $\mathcal{C}(H,\dot{X},\Lambda)$ represents a matrix ellipsoid under certain conditions on the dataset, as asserted in the following lemma.

\begin{lemma} \label{eq:lemma_ellipsoid}
    The set $\mathcal{C}(H,\dot{X},\Lambda)$ defined in \eqref{eq:ASE} is a matrix ellipsoid if $H \Lambda H^{\top} \succ 0$.
\end{lemma}

{For the underlying system \eqref{eq:sys_real}, Lemma \ref{eq:lemma_ellipsoid} provides a sufficient condition which ensures that the true system parameter $(\mathcal{A},\mathcal{B})$ lies in a bounded set.} Consequently, given the initial state and the input sequence, the future state of the system can be characterized as bounded. A similar result can be found in the context of behavioral systems approaches, namely, the renowned fundamental lemma \cite{willemsNotePersistencyExcitation2005}, which provides conditions such that the future trajectory can be uniquely determined by the initial trajectory and the input sequence. From a systems theory perspective, both approaches address the problem of identifiability conditions on the data and have been shown to be equivalent when noise-free data is considered \cite{markovskyBehavioralSystemsTheory2021}.

\begin{remark}
	In this study, we focus on systems with disturbances, hence we posit the assumption that $\delta > 0$. However, when we consider noise-free systems by setting $\delta = 0$ and {$\mathcal{W}=\{0\}$}, we have
	\begin{equation} \label{eq:noise_free_remark1}
		\begin{bmatrix}
			I \\ \Delta^{\top}
		\end{bmatrix}^{\top}\!
		\Xi(H,\dot{X},\Lambda)
		\begin{bmatrix}
			I \\ \Delta^{\top}
		\end{bmatrix}
		= 0
	\end{equation}
	holds for any admissible system $\Delta$.
	Therefore, we can conclude $\Delta E_e \Delta^{\top} + \Delta F_e + (\Delta F_e)^{\top} + {G_e} = 0$, which is equivalent to $(\Delta + F_e^{\top}E_e^{-1})E_e(\Delta + F_e^{\top}E_e^{-1})^{\top} = 0$. Recall that $E_e \succ 0$. This implies that \eqref{eq:noise_free_remark1} has only one solution $\Delta = -F_e^{\top}E_e^{-1}$, which is equivalent to $\Delta_r$ given that $\dot{X}=\Delta_r H$.
\end{remark}

For brevity, we define $\mathcal{C}(H,\dot{X},\Lambda)$ that satisfies $H \Lambda H^{\top} \succ 0$ as an \emph{admissible system ellipsoid} (ASE) of \eqref{eq:sys_real}.
\hspace{-1ex}
\subsection{Adaptive Tube-based Data-driven Predictive Control} \label{subsec:ATDPC}
In adaptive tube-based data-driven predictive control (ATDPC), the stage cost function is defined as
\begin{equation} \label{eq:stage_cost}
	l(x,u):=x^{\!\top}Qx+u^{\!\top}Ru,
\end{equation}
where $Q \!\in\! \mathbb{R}^{n_x \times n_x}$, $\!Q \!\succ\! 0$ and $R\!\in\! \mathbb{R}^{n_u \times n_u}$, $\!R \!\succ\! 0$. The terminal cost function is defined as $V_f(x) := x^{\top} P_{T} x$ where $P_T \in \mathbb{R}^{n_x \times n_x}$, $\!P_T\succ 0$. The corresponding terminal constraint set is $\mathcal{X}_T := \left\{ x \mid x^{\top} P_{T} x \leq L_T \right\}$ where $L_T > 0$. According to the classical framework of tube-based model predictive control, we design a nominal system as
\begin{equation} \label{nominal_sys}
	\bar{x}(k+1) = \bar{A}(k) \bar{x}(k) + \bar{B}(k) \bar{u}(k),
\end{equation}
where $\bar{\Delta}(k)\!:=\![ \bar{A}(k) \ \bar{B}(k)]$ is the center of $\mathcal{C}^{\star}_{k}$, $\bar{x}(k) \in \mathbb{R}^{n_x}$ and $\bar{u}(k) \in \mathbb{R}^{n_u}$ are the state and input of the nominal system. {Note that $\mathcal{C}^{\star}_{k},k \in \mathbb{Z}_{\geq 0}$ denotes a sequence of ASEs learned from closed-loop data samples at time instant $k$. Each set $\mathcal{C}^{\star}_{k}$ contains the true system parameters $(\mathcal{A}, \mathcal{B})$, and its parameters are to be designed in Section \ref{subsec:learn_Cs}.}

To measure the difference between the nominal and the true systems, we define the error between their states as $e(k):=x(k)-\bar{x}(k)$. For a matrix $K \in \mathbb{R}^{n_x\times n_u}$, the control input applied to \eqref{eq:sys_real} is designed as
\begin{equation} \label{input}
	u(k)=K e(k)+\bar{u}(k),
\end{equation}
where the term $Ke(k)$ is designed to reduce the error between the true system and the nominal system. The other term, $\bar{u}(k)$, is designed to control the nominal system to meet the control specifications, constraints, and uncertainties. The value of $\bar{u}(k)$ is obtained by solving:
\begin{equation}  \label{eq:TDPC_OPT}
	\setlength{\abovedisplayskip}{3pt}
	\setlength{\belowdisplayskip}{3pt}
	\begin{split}
		&\min_{\bar{u}_{\text{seq},k}} V_s(\bar{x}_{0|k},\bar{u}_{\text{seq},k}) \!=\!{\textstyle\sum\nolimits_{i=0}^{N-1} l(\bar{x}_{i|k},\bar{u}_{i|k}) + V_f(\bar{x}_{N|k})}\\[-2pt]
		& \ \text{s.t.}\\
		& \left\{
		\begin{array}{lc}
			\bar{x}_{i+1|k} = \bar{A}(k)\bar{x}_{i|k} + \bar{B}(k)\bar{u}_{i|k},\\
			\bar{x}_{i|k} \in \mathcal{X} \ominus E_{i|k},\\
			\bar{u}_{i|k} \in \mathcal{U} \ominus K E_{i|k},\\
			\bar{x}_{N|k} \in \mathcal{X}_T \ominus E_{N|k},\\
			\bar{x}_{0|k} = x(k), \ E_{0|k} = e(k), \ i \in \mathbb{Z}_{[0,N-1]},
		\end{array}
		\right.
	\end{split}
\end{equation}
where $\bar{x}_{0|k}:=\bar{x}(k),\bar{u}_{\text{seq},k}:=(\bar{u}_{0|k},\bar{u}_{1|k},\dots,\bar{u}_{N-1|k})$, $\mathcal{X}$ and $\mathcal{U}$ are convex sets that contain the corresponding origin points and restrict the states and inputs, respectively, and $E_{i|k}$ with $k\in \mathbb{Z}_{\geq 0}$ and $i\in \mathbb{Z}_{[0,N]}$ is the convex set of the reachable error $e(i+k)$ computed at instant $k$, i.e., $e(i+k) \in E_{i|k}$. Then, since the objective function is a convex function, the problem in \eqref{eq:TDPC_OPT} is a convex optimization problem\cite{farjadniaRobustDatadrivenPredictive2023}.\footnote{The optimization problem can be solved through standard solvers, e.g., the MATLAB function $\mathtt{fmincon}$.} We pick $\bar{u}(k)$ as the optimal $\bar{u}_{0|k}$. The design of $K$, $\mathcal{C}_k^{\star}$, $E_{i|k}$, and $L_T$ will be investigated in Section \ref{sec:Tube-based Data-driven Predictive Control}.

\section{Problem Formulation} \label{sec:Problem Formulation}

Due to the uncertainty introduced by $w(k)$, the real system model cannot be obtained accurately from the dataset. As an alternative, it is easy to observe that $(\mathcal{A},\mathcal{B}) \in \mathcal{C}(H,\dot{X},\Lambda)$. Therefore, the data-driven analysis of the true system turns to the set-valued analysis of $\mathcal{C}(H,\dot{X},\Lambda)$ and a smaller ASE implies less uncertainty about the true system. Considering the tradeoff between performance and robustness, it is easier to find a feasible controller to handle the shrunken uncertainty with potentially improved performance guarantee{\cite{2021_Trade-off}}. Since the baseline requirement is to effectively learn a controller to stabilize all systems consistent with the data, we aim to design learning strategies that minimize the volume of the ASE through actively generating and selecting data samples.

When the controller is not implemented, we can employ open-loop active learning strategies that design inputs based on an appropriate metric corresponding to the volume of the ASE. {It is noted that the metric relies on $\dot{X}$, which cannot be directly designed and is not fully accessible when the input is being designed. {Specifically, since $x_s(d_n + 1)$ is not available before applying $u_s(d_n)$ to the system and measuring the resulting state, the entire $\dot{X}$ cannot be used in the design of $u_s(d_n)$.} Therefore, to enable open-loop active learning strategies, the design of inputs needs to be performed using only $H$ and $\Lambda$. {However, the available information collected in $H$ is not fully exploited for input design. In particular, random inputs are usually applied in open-loop experiments to facilitate data-driven control \cite{depersis2020formulas,dorflerBridgingDirectIndirect2023}.} To mitigate the effect of the random inputs to the open-loop system, the random inputs are required to be selected from a predefined compact and convex set $\mathcal{U}_o$, which can be designed based on safety and stability considerations for a particular engineering application.}

However, when the controller is implemented, the system inputs are determined by the controller and cannot be arbitrarily designed. {Thus, we consider an ATDPC with a closed-loop active learning strategy in this work to reduce the size of the learned ASE by selecting informative closed-loop data samples.} With the statements above, we formally summarize the problems to be solved as follow.
\begin{problem} \label{prob:measure}
	Find a method to quantify the volume of ASEs {using historical state-input trajectory without the system model}, i.e., find a function $\hat{\mu}(H,\Lambda)$ that can characterize the volume of $\mathcal{C}(H,\dot{X},\Lambda)$.
\end{problem}
\begin{problem} \label{prob:active_learning}
	{Given $\hat{\mu}(\cdot)$,} find a method to actively design $u_s(i) \!\in\! \mathcal{U}_o$, $i\in\mathbb{Z}_{[0,t_0-1]}$ and $\lambda_i\geq0$, $i\in\mathbb{Z}_{[1,n]}$ such that $\hat{\mu}(H,\Lambda)$ is minimized.
\end{problem}
\begin{problem} \label{prob:adaptive_controller}
    Once $\hat{\mu}(H,\Lambda)$ is sufficiently small, design an ATDPC as in \eqref{input} and \eqref{eq:TDPC_OPT} that could stabilize the system \eqref{eq:sys_real} and effectively learn from closed-loop data using an active data {selection strategy}.
\end{problem}

\section{Measure and Overapproximation of ASEs} \label{sec:MoASE}
{This section focuses on the first problem as posed above in Section \ref{sec:Problem Formulation}. While Lemma \ref{Lem:MoME} provides a feasible solution to quantify the volume of ASEs, it requires knowledge of $\dot{X}$. To address this issue, the problem is divided into two parts based on the number of data samples. First, we analyze the volume of ASEs with a square data matrix. An effective method that could characterize the accurate volume of the ASE by $H$ and $\Lambda$ is presented. Furthermore, we extend this result to general ASEs and present an effective method to characterize the upper bound of the volume of ASEs.}
\subsection{Measure of ASEs with Square Data Matrix}
In this section, we focus on ASEs with $n=n_h$, and consider
\begin{equation} \label{eq:ASE_vol_hat}
	\hat{\mu}(H,\Lambda):=\beta(\mathrm{tr}(\Lambda)\delta)^{n_x n_h/2}\det(H\Lambda H^{\top})^{-n_x/2},
\end{equation}
where $\beta$ is a constant that depends only on $n_x$ and $n_u$, as defined in Lemma \ref{Lem:MoME}. We prove that under special cases, $\hat{\mu}(H,\Lambda)$ is a precise representation of the volume of $\mathcal{C}(H,\dot{X},\Lambda)$ in the following lemma.

\begin{lemma} \label{lem:vol_squareC}
    Let the order of the uncertain system be fixed, i.e., $n_x$ and $n_u$ are known and constant parameter. For an ASE $\mathcal{C}(H,\dot{X},\Lambda)$, if number of data samples satisfies $n=n_h$, then $\mu (\mathcal{C}(H,\dot{X},\Lambda)) = \hat{\mu}(H,\Lambda)$.
\end{lemma}
\begin{proof}
    {Based on Lemma \ref{Lem:MoME}}, $\mu(\mathcal{C}(H,\dot{X},\Lambda))$ is equivalent to
    \begin{equation} \label{eq:squareC_cal}
            \beta
            \det(F_e^{\top}E_e^{-1}F_e - G_e)^{{n_h}/{2}} \det(E_e^{-1})^{n_x/2},
    \end{equation}
    where $E_e$, $F_e$, and $G_e$ are defined in \eqref{eq:def_ellip_matrix} in the appendix. Note that $H \Lambda H^{\!\top}\!\!\succ\! 0$ and $H$ is a square matrix. We can obtain that $H$ is a square matrix with full rank. Thus, we have {$H^{\top}\!E_e^{-1}H\!=\!\Lambda^{-1}$}, which implies $F_e^{\top}E_e^{-1}F_e - G_e\!=\!\mathrm{tr}(\Lambda)\delta I_{n_x}$. Then, the conclusion follows from the fact that $\det(\mathrm{tr}(\Lambda) \delta I_{n_x}) = (\mathrm{tr}(\Lambda) \delta)^{n_x}$.
\end{proof}

{
\begin{remark} \label{rem:H_full_rank}
	The condition of having a full rank data matrix $H$, i.e. $\operatorname{rank}(H)=n_h$, in Lemma \ref{lem:vol_squareC} is a standard assumption in data-driven control literature \cite{depersis2020formulas,2021_Trade-off,alanwarDataDrivenReachabilityAnalysis2023}. It ensures that the data encode all the information for the direct design of control laws \cite{depersis2020formulas} or for the characterization of a bounded set of admissible systems \cite{2021_Trade-off,alanwarDataDrivenReachabilityAnalysis2023}. This condition can be easily verified given the data set. For the noisy systems discussed in this work, the full rank condition holds with high probability due to the presence of random disturbances \cite{muller2023inputdesign}.
\end{remark}
}

Also note that $\det(H\Lambda H^{\top})=\det(\Lambda)\det(H)^2$ when $H$ is a square matrix. Then, we can reformulate $\hat{\mu}(H,\Lambda)$ as
\begin{equation}
	\hat{\mu}(H,\Lambda) \!=\!  \beta(\mathrm{tr}(\Lambda)\delta)^{{n_x n_h}\!/{2}}\det(\Lambda)^{-{n_x}\!/{2}}\left| \det(H) \right|^{-n_x}\!\!.
\end{equation}
{In addition, we continue to discuss a more general case where $H$ is not a square matrix in the following subsection.}
\subsection{Overapproximation of General ASEs}
{Lemma \ref{lem:vol_squareC} allows us to precisely evaluate the volume of $\mathcal{C}(H,\dot{X},\Lambda)$ by $\hat{\mu}(H,\Lambda)$ when $n = n_h$.  However, it is challenging to extend this result to general ASEs, as $HE_e^{-1}H^{\top}\!=\!\Lambda^{-1}$ no longer holds when $n\geq n_h$. Consequently, $F_e^{\top}E_e^{-1}F_e - G_e$ is not available since $W$ is unknown. Nevertheless, we show in the following lemma that $\hat{\mu}(H,\Lambda)$ serves as an upper bound on the volume of the ASE in this case.

\begin{lemma} \label{lem:vol_generalC}
    For any ASE $\mathcal{C}(H,\dot{X},\Lambda)$, the upper bound of the volume can be given as ${\mu}(\mathcal{C}(H,\dot{X},\Lambda)) \leq \hat{\mu}(H,\Lambda)$.
\end{lemma}
{\begin{proof}
    We define $\mathbf{Q}$ and $\mathbf{Q}_p$ as
    \begin{equation*}
    	\begin{split}
    		\mathbf{Q}:=& F_e^{\top}E_e^{-1}F_e - G_e \\
    		=& \dot{X} \Lambda H^{\top}\!(H \Lambda H^{\top}\!)^{-1} H \Lambda \dot{X}^{\top}\!-\dot{X} \Lambda \dot{X}^{\top}\!+\operatorname{tr}(\Lambda) \delta I, \\
    		\mathbf{Q}_p := & (H \Lambda^{1 / 2})^{\top}(H \Lambda^{1 / 2} \Lambda^{1 / 2} H^{\top})^{-1} H \Lambda^{1 / 2}.
    	\end{split}
    \end{equation*}
    Since $\mathbf{Q}_p$ is a projection matrix, i.e., $\mathbf{Q}_{p}^{2}=\mathbf{Q}_{p}$, we have
    \begin{equation} \label{eq:lem_Gvol_1}
    	\mathbf{Q}=\dot{X} \Lambda^{1 / 2} \mathbf{Q}_{p} \Lambda^{1 / 2} \dot{X}^{\top}\!-\dot{X} \Lambda^{1 / 2} \Lambda^{1 / 2} \dot{X}^{\top}\!+\operatorname{tr}(\Lambda) \delta I.
    \end{equation}
    {With \eqref{eq:matrix_sys}, we have $\mathbf{Q}_{p} \Lambda^{1 / 2} \dot{X}^{\top}-\Lambda^{1 / 2} \dot{X}^{\top}=(\mathbf{Q}_{p}-I) \Lambda^{1 / 2} W^{\top}$ as in the proof of \cite[Lemma 1]{bisoffi2022data}, which yields
    \begin{equation*}
    	\mathbf{Q}=W \Lambda^{1 / 2}(\mathbf{Q}_{p}-I) \Lambda^{1 / 2} W^{\top}\!+\operatorname{tr}(\Lambda) \delta I \preceq \operatorname{tr}(\Lambda) \delta I
    \end{equation*}
    using \eqref{eq:lem_Gvol_1}.} {Furthermore, we have $\mathbf{Q} \!\succ\! 0$ since $\mathcal{C}(H,\dot{X},\Lambda)$ is an ASE. Additionally, since $0 \!\prec\! \mathbf{Q} \!\preceq\! \operatorname{tr}(\Lambda)\delta I$ holds regardless of the selection of $W$, we obtain $\det(\mathbf{Q}) \leq \det(\operatorname{tr}(\Lambda) \delta I_{n_x}) = (\operatorname{tr}(\Lambda) \delta)^{n_x}$ by invoking the matrix determinant lemma.} Thus, $\mu(\mathcal{C}(H, \dot{X}, \Lambda))\! \!=\!\! \beta
    	\det(\mathbf{Q})^{{n_h}/{2}} \det(E_e)^{-n_x/2} \!\leq\! \hat{\mu}(H,\Lambda)$ holds according to Lemma \ref{Lem:MoME}.
\end{proof}}

\section{Open-loop Active Learning for Data-driven Modeling}
\label{sec:Active Online Learning for Modeling}
{In this section, we turn to the second problem in Section \ref{sec:Problem Formulation}. The design of {open-loop active learning} is divided into two different phases, i.e., the phase of formation and the phase of contraction. The phase of formation aims to generate $n_h$ data samples and form an ASE with small volume based on them. The phase of contraction aims to contract the volume of ASEs by properly designing the input.}

Lemmas \ref{lem:vol_squareC} and \ref{lem:vol_generalC} indicate that the volume of ASE and the determinant of the data matrix are closely related. However, the volume of ASE depends on the unknown system. Thus, it is challenging to design a globally optimal input design strategy, which is demonstrated through the following example.

{{\begin{example}
    Consider the condition where $n_x\!=\!n_u\!=\!1$, $\delta\!=\!1$, and the open-loop input constraint set $\mathcal{U}_o=\{u\in\mathbb{R} \ | \ |u| \leq 1\}$. Assume that the initial state $x(0)=1$, and two data samples need to be collected. We aim to actively design $u(i)\in \mathcal{U},\lambda_{i+1}>0,i\in\{0,1\}$ to minimize the volume of ASE. Suppose that the data matrices are collected as
    \begin{equation*}
    	H=
    	\begin{bmatrix}
    		1 & x(1) \\
    		u(0) & u(1)
    	\end{bmatrix}, \
    	\dot{X}=[x(1) \ x(2)], \ \Lambda=\text{diag}(\lambda_1,\lambda_2),
    \end{equation*}
    where $x(i+1)= \mathcal{A}x(i)+\mathcal{B}u(i), i\in\{0,1\}$. With Lemma \ref{lem:vol_squareC}, we have $\mu(\mathcal{C}(H,\dot{X},\Lambda)) = \beta(\lambda_1\!+\!\lambda_2)(\lambda_1\lambda_2)^{-\frac{1}{2}}\left| \det(H) \right|^{-1}.$
    By replacing the term $\det(H)$ with
    $$\det(H) = u(1) - u(0)x(1) = u(1)-\mathcal{A}u(0)-\mathcal{B}u(0)^2,$$
    it can be concluded that the volume is a function of $u(0)$, $u(1)$, $\lambda_1$, $\lambda_2$ and the unknown system matrices $(\mathcal{A},\mathcal{B})$. Note that the volume is minimized when $\left| \det(H) \right|$ is maximized. If {$\mathcal{A}=1$ and $\mathcal{B}=5$}, then $\left| \det(H) \right| = |u(1) - u(0) -5u(0)^2|$ is maximized when $\{u(0),u(1)\}=\{1,-1\}$. However, if {$\mathcal{A}=-1$ and $\mathcal{B}=-5$}, we have $\left| \det(H) \right|= |u(1) + u(0) + 5u(0)^2|$, which is maximized when $\{u(0),u(1)\}=\{1,1\}$.
\end{example}}}

Therefore, we will employ {open-loop active learning} based on a suboptimal greedy strategy to optimize the control input at each instant in the {open-loop stage}.

\subsection{Learning for An ASE with Small Volume} \label{subsec:formation}
In this subsection, we focus on the active input design problem to generate an ASE with {finite} and relatively small volume within $n_h$ data samples. {Note that the volume of the ASE can be computed by $\hat{\mu}(H,\Lambda)$ according to Lemma \ref{lem:vol_squareC}. We first discuss the design of $\lambda_i\geq0$, $i\in\mathbb{Z}_{[1,n_h]}$ to minimize $\hat{\mu}(H,\Lambda)$ in the following lemma.
\begin{lemma} \label{lem:lambda_square}
	Suppose that $n \!=\! n_h$, $\Lambda \!\in\! \mathcal{D}_{n_h}$, and $H \Lambda H^{\top} \!\succ\! 0$. Then, we have $\hat{\mu}(H,\Lambda) \geq \hat{\mu}(H,I_{n_h})$.
\end{lemma}
\begin{proof}
	Since $H \Lambda H^{\top} \succ 0$, we can obtain that $H$ has full row rank. Besides, $H$ is a square matrix when $n=n_h$. This implies that $H$ also has full column rank and $\Lambda \succ 0$. Note that $\Lambda=\text{diag}(\lambda_1,\lambda_2,\dots,\lambda_{n_h})$. We can conclude that $\lambda_i > 0, \forall i \in \mathbb{Z}_{[1,n_h]}$. Next, we have
	\begin{equation} \label{eq:lambda_square_p1}
		\hat{\mu}(H,\Lambda)=\tfrac{\beta(\delta\sum_{i=1}^{n_h}\lambda_i)^{n_x n_h / 2}}{(\prod_{i=1}^{n_h}\lambda_i)^{n_x / 2} \left| \det(H) \right|^{n_x}}.
	\end{equation}
	By invoking the mean inequality chain, we have
	$
		\textstyle\sum\nolimits_{i=1}^{n_h}\lambda_i \geq n_h (\textstyle\prod\nolimits_{i=1}^{n_h}\lambda_i)^{1 / n_h}.
	$
	Then, combined with \eqref{eq:lambda_square_p1}, it can be concluded that
	$
		\hat{\mu}(H,\Lambda) \geq \beta(n_h \delta )^{n_x n_h / 2}\left| \det(H) \right|^{-n_x},
	$
	which is equivalent to $\hat{\mu}(H,I_{n_h})$.
\end{proof}

Lemma \ref{lem:lambda_square} implies that the optimal choice of $\Lambda$ that enables an ASE when $n=n_h$ is $I_{n_h}$. We are in a position to investigate the {active learning strategy which designs} input signals. Recall that $\hat{\mu}(H,\Lambda)$ depends on $\left| \det(H) \right|$. A useful lemma that can be used to compute the absolute value of the matrix determinant is introduced as follows.
\begin{lemma}\label{lem:det_cal}
    For a full rank matrix ${D} \in \mathbb{R}^{n \times n}$, the absolute value of the determinant of ${D}$ can be calculated by
    \begin{equation} \label{eq:det_cal}
            \!\Vert \hspace{-1pt}{D}_{[1]}\hspace{-1pt} \Vert_{\hspace{-1pt}2}\hspace{-3.6pt}\prod\limits_{i=1}^{n-1} \hspace{-3.6pt}\Vert \hspace{-1pt}{D}_{[i+\hspace{-.6pt}1]} \!-\! {D}_{[1:i]}({D}_{[1:i]}^{\top}{D}_{[1:i]})^{\!-1}\!{D}_{[1:i]}^{\!\top}{D}_{[i+\hspace{-.6pt}1]} \hspace{-1pt}\Vert_{\hspace{-1pt}2}.
    \end{equation}
\end{lemma}
\vspace{1ex}
\begin{proof}
	{Suppose that the QR factorization of $D$ is given by $D=QR$, where $Q$ is an orthogonal matrix and $R$ is an upper triangle matrix. Then, it follows that $\left| \det(D) \right| = \left| \det(Q) \right|\left| \det(R) \right|$. Given that $Q$ is an orthogonal matrix we have $\left| \det(Q) \right| = 1$. Besides, considering that $R$ is an upper triangle matrix, we also have $\left| \det(R) \right|=\prod_{i=1}^{n}|R_{[i,i]}|$.
	
	Note that the first term at the right-hand side is $\Vert \hspace{-1pt}{D}_{[1]}\hspace{-1pt} \Vert_{\hspace{-1pt}2}$, which is equal to $|R_{[1,1]}|$ as computed in the Gram-Schmidt factorization algorithm presented in \cite[Page 254]{golub2013matrix}. Each subsequent term in (21) can be presented as $$\Vert \hspace{-1pt}{D}_{[i+\hspace{-.6pt}1]} \!-\! {D}_{[1:i]}({D}_{[1:i]}^{\top}{D}_{[1:i]})^{\!-1}\!{D}_{[1:i]}^{\!\top}{D}_{[i+\hspace{-.6pt}1]} \hspace{-1pt}\Vert_{\hspace{-1pt}2},$$ which is equivalent to $$\Vert \hspace{-1pt}(I- {D}_{[1:i]}({D}_{[1:i]}^{\top}{D}_{[1:i]})^{\!-1}\!{D}_{[1:i]}^{\!\top}){D}_{[i+\hspace{-.6pt}1]} \hspace{-1pt}\Vert_{\hspace{-1pt}2}, i \in \mathbb{Z}_{[1,n-1]}.$$
	
	For simplicity, we define $E_i:=[I_i^{\top} \ 0_{n \times i}^{\top}]^{\top}$. Then, by replacing $D_{[1:i]}$ with $QRE_i$, we have
	\begin{equation*}
		{D}_{[1:i]}({D}_{[1:i]}^{\top}{D}_{[1:i]})^{\!-1}\!{D}_{[1:i]}^{\!\top} = QRE_i (E_i^{\top}R^{\top}RE_i)^{-1}E_i^{\top}R^{\top}Q^{\top}.
	\end{equation*}
	Since $R$ is an upper triangle matrix, $RE_i$ can be partitioned as $[R_1^{\top} \ 0^{\top}]^{\top}$ where $R_1$ is also upper triangular. Thus, by replacing $RE_i$ with $[R_1^{\top} \ 0^{\top}]^{\top}$, we have
	\begin{align*}
		{D}_{[1:i]}({D}_{[1:i]}^{\top}{D}_{[1:i]})^{\!-1}\!{D}_{[1:i]}^{\!\top}
		&= Q \begin{bmatrix}
			R_1 \\ 0
		\end{bmatrix}
		R_1^{-1}R_1^{-T}
		\begin{bmatrix}
			R_1 \\ 0
		\end{bmatrix}^{\top}\!\!Q^{\top} \\
		&= QE_iE_i^{\top}Q^{\top} \\
		&= Q_{[1:i]}Q_{[1:i]}^{\top}.
	\end{align*}
	According to the Gram-Schmidt factorization algorithm in \cite[Page 254]{golub2013matrix}, this implies
	\begin{equation*}
		|R_{[i+1,i+1]}| = \big\Vert(I - Q_{[1:i]}Q_{[1:i]}^{\top})D_{[i+1]}\big\Vert_2, i \in \mathbb{Z}_{[1,n-1]}.
	\end{equation*}
	Therefore, we obtain$$\Vert \hspace{-1pt}(I- {D}_{[1:i]}({D}_{[1:i]}^{\top}{D}_{[1:i]})^{\!-1}\!{D}_{[1:i]}^{\!\top}){D}_{[i+\hspace{-.6pt}1]} \hspace{-1pt}\Vert_{\hspace{-1pt}2} = |R_{[i+1,i+1]}|$$
	for all $i \in \mathbb{Z}_{[1,n-1]}$. Finally, since $|\det{D}| = |\det{Q}| |\det{R}| = \prod_{i=1}^{n}|R_{[i,i]}|$, we can conclude that (21) holds as claimed.}
\end{proof}
}

Note that $H$ is generated column by column, which coincides with the computation order in Lemma\hspace{2pt}\ref{lem:det_cal}. {Therefore, to obtain a square matrix $\hspace{-1pt}H\hspace{-1pt}$ with a large $\hspace{-1pt}|\hspace{-1pt}\det(H)\hspace{-.5pt}|$, we propose an active learning strategy by invoking Lemma \ref{lem:det_cal} and greedily maximizing the Euclidean norm of each entry of \eqref{eq:det_cal}. Specifically, $H$ corresponds to $D_{[1:i]}$ in \eqref{eq:det_cal}, which represents the data collected in the form of columns. Suppose that $x$ is the most recently collected state and $u$ is the input to be designed. Then, $[x^{\!\top} \ u^{\!\top}]^{\!\top}\!\!$ corresponds to $D_{[i+1]}$ in \eqref{eq:det_cal}, which represents the next column of data to be collected by $H$.} The input signal is designed as
\begin{equation}  \label{eq:formation_opt}
    u^{\star} \!\!=\! f_{\text{F}}(H,x)\!:=\!\mathop{\arg\max}\limits_{u \in \mathcal{U}_o}
    \underbrace{\Big{\Vert}\!\!\begin{bmatrix}
    		x \\[-.9ex] u
    	\end{bmatrix}
    	\!-\! H(H^{\!\top} \!H)^{-1}H^{\!\top}\!
    	\begin{bmatrix}
    		x \\[-.9ex] u
    	\end{bmatrix} \!\! \Big{\Vert}_2}_{{:=J_{\text{F}}(H,x,u)}},
\end{equation}
{where an input constraint $u\in \mathcal{U}_o$ is considered to reflect the restrictions on input signals in practical applications. Note that $\mathcal{U}_o$ is assumed to be a compact and convex set, and that $J_{\text{F}}(H,x,u)$ is equivalent to $\Vert
	(I-H(H^{\!\top} \!H)^{-1}H^{\!\top})\left[x^{\!\top} \  u^{\!\top}\right]^{\!\top}\!\!
	\Vert_2$
where $I\!-\!H(H^{\!\top} \!H)^{-1}H^{\!\top} \!\succeq\!0$ since $H(H^{\!\top} \!H)^{-1}H^{\!\top}$ is a projection matrix. The problem in \eqref{eq:formation_opt} is a concave quadratic programming problem, which can be solved using the algorithm in \cite[Section 3]{yeAffineScalingAlgorithms1992} when $\mathcal{U}_o$ is an ellipsoid, or the algorithm in \cite[Section IV]{zwartGlobalMaximizationConvex1974} when $\mathcal{U}_o$ is a polytope. For general cases, $\mathcal{U}_o$ can be approximated by a polytope, and the problem can be solved accordingly \cite{bronsteinApproximationConvexSets2008}.} If $J_{\text{F}}(H,x,f_{\text{F}}(H,x))\!=\!0$, {a random input from $\mathcal{U}_o$} will be applied without being collected, thus preventing data matrix $H$ from absorbing non-innovative data. {Since we aim to minimize $\hat{\mu}(H,\Lambda)$, we define non-innovative data as data samples that, when collected, would undesirably increase the value of $\hat{\mu}(H,\Lambda)$.} {For clarity, the pseudo-code of the phase of formation is summarized in Algorithm \ref{alg:formation}.}
\begin{algorithm}
	\caption{Phase of Formation}
	\label{alg:formation}
	\begin{algorithmic}[1]
		\STATE initialize $x, i\leftarrow0$
            \STATE $u \leftarrow \mathop{\arg\max}\limits_{u \in \mathcal{U}_o} \Vert u \Vert_2$ and update $h$
            \STATE apply $u$ to \eqref{eq:sys_real} and measure the response $x$
            \STATE collect data sample $H \leftarrow h$, $\dot{X} \leftarrow x$
		\STATE \textbf{while }$i \leq n_h -1$
		\STATE \hspace{1em} $u \leftarrow f_{\text{F}}(H,x)$ and update $h$
            \STATE \hspace{1em} \textbf{if} $J_{\text{F}}(H,x,u) > 0$
            \STATE \hspace{2em} apply $u$ to \eqref{eq:sys_real} and measure the response $x$
            \STATE \hspace{2em} collect data sample $H \leftarrow [H \ h]$, $ \dot{X} \leftarrow [\dot{X} \ x]$
            \STATE \hspace{2em} set $i \leftarrow i + 1$
            \STATE \hspace{1em} \textbf{else}
            \STATE \hspace{2em} $u \leftarrow \text{a random element in } \mathcal{U}_o$
            \STATE \hspace{2em} apply $u$ to \eqref{eq:sys_real} and measure the response $x$
            \STATE \hspace{1em} \textbf{end if}
		\STATE \textbf{end while}
		\STATE \textbf{return} $H, \dot{X}$
	\end{algorithmic}
\end{algorithm}

\begin{remark} \label{rem:algo_full_rank}
    Note that the additional learning criterion ensures that each entry in the right hand of \eqref{eq:det_cal} is none-zero. As such, Lemma \ref{lem:det_cal} asserts that $\det(H) \neq 0$. Then the proposed active learning strategy guarantees that a full rank data matrix $H$ will be obtained after the phase of formation. {In particular, when we consider a noise-free controllable system and replace \eqref{eq:formation_opt} with an optimization problem searching for a positive solution, the proposed strategy degenerates to the online experiment design method introduced in \cite{2022_Beyond}.} Namely, any $u$ that makes $J_{\text{F}}(H,x,u) \neq 0$ is a feasible solution for the experiment design problem in \cite{2022_Beyond} and \emph{vice versa}.
\end{remark}

\subsection{Learning for ASEs with Decreasing Bound of Volume}
Following the framework of the Algorithm \ref{alg:formation}, we will obtain a matrix ellipsoid which contains all admissible systems. However, the volume of the matrix ellipsoid can still be too large to obtain a stable controller. {In what follows, an active learning strategy is proposed to decrease the upper bound of the volume of ASEs with the help of Lemma \ref{lem:vol_generalC}.}
\begin{lemma}[\text{\hspace{-.01em}\cite[Lemma 1.1]{2007_MDL}}] \label{lem:MDL}
    Suppose that ${D}$ is an invertible matrix and $u,v$ are column vectors with suitable dimension. Then, we have $\det({D}+uv^{\!\top}\!) \!=\! ( 1+v^{\!\top}\!{D}^{-1}u) \det({D})$.
\end{lemma}

In the phase of contraction, the purpose is to design an active learning algorithm that maximizes $\det(HH^{\top})$ at each step. Similar to the phase of formation, an input $u$ needs to be designed given the data matrix $H$ formed by historical data and the current state $x$. {Note that after the phase of formation, $\Lambda$ is designed as $I_{n_h}$ according to Lemma \ref{lem:lambda_square}, and $H$ is a full rank square matrix as stated in Remark \ref{rem:algo_full_rank}. Therefore, we have $H \Lambda H^{\top} \succ 0$ after the phase of formation.} Then, a necessary and sufficient condition for the contraction of the upper bound of the volume of ASEs is presented in the next theorem.
{\begin{theorem} \label{thm:ugen_contraction}
    Suppose that $\mathcal{C}(H,\dot{X},\Lambda)$ is an ASE where $H \in \mathbb{R}^{n_h \times n}, \dot{X} \in \mathbb{R}^{n_x \times n}$. Denote the current state by $x\in \mathbb{R}^{n_x}$ and the designed input by $u \in \mathbb{R}^{n_u}$. Then there exist $\lambda>0$ and $u\in \mathcal{U}_o$ such that $\hat{\mu}(H,\Lambda)>\hat{\mu}(H_+,\Lambda_+)$ if and only if there exists $u \in \mathcal{U}_o$ such that
    \begin{equation} \label{eq:contract_criteria}
    h^{\top}\!( H \Lambda H^{\top})^{-1}	h > {n_h} / {\mathrm{tr}(\Lambda)},
    \end{equation}
    where $h = [x^{\top} \ u^{\top}]^{\top}$, $H_{+} := [H \ h]$, and $\Lambda_{+} := \text{diag}(\Lambda,\lambda)$.
\end{theorem}
\begin{proof}
    {Since $\mathcal{C}(H,\dot{X},\Lambda)$ is an ASE, we have $H\Lambda H^{\top} \succ 0$. Also note that $H_+\Lambda_+ H^{\top}_+ = H\Lambda H^{\top} + \lambda h h^{\top}$ and $\lambda>0$. We can obtain $H_+\Lambda_+ H^{\top}_+ \succ 0$.} {We can also conclude that $\beta$ in $\hat{\mu}(H,\Lambda)$ and $\hat{\mu}(H_+,\Lambda_+)$ are the same, since $\beta$ depends only on the dimensions of the elements in the ASE according to Lemma \ref{Lem:MoME}, and the dimensions of the elements in the ASEs corresponding to $\hat{\mu}(H,\Lambda)$ and $\hat{\mu}(H_+,\Lambda_+)$ are identical.} Then, given that $\mathcal{C}(H,\dot{X},\Lambda)$ is an ASE, $\hat{\mu}(H,\Lambda)>\hat{\mu}(H_+,\Lambda_+)$ is equivalent to
    \begin{equation} \label{eq:contract_compare1}
    	f_p:=\frac{\mathrm{tr}(\Lambda_+)^{n_x n_h / 2}\det(H\Lambda H^{\top})^{n_x /2}}{\mathrm{tr}(\Lambda)^{n_x n_h / 2}\det(H_+\Lambda_+ H_+^{\top})^{n_x /2}}<1
    \end{equation}
    by expanding the terms $\hat{\mu}(H,\Lambda)$ and $\hat{\mu}(H_+,\Lambda_+)$ according to the definition \eqref{eq:ASE_vol_hat}. Recall that $H\Lambda H^{\top} \succ 0$. We have
     \begin{equation} \label{contract_det_ite}
    	\begin{aligned}
    		\det(H_{+}\Lambda_+ H_{+}^{\top}) = ( 1 + \lambda m) \det( H\Lambda H^{\top})
    	\end{aligned}
    \end{equation}
    by invoking Lemma \ref{lem:MDL} and denoting $m:=h^{\top}( H \Lambda H^{\top})^{-1} h$. In addition, we also have $\mathrm{tr}(\Lambda_+)=\mathrm{tr}(\Lambda) + \lambda$. Combined with \eqref{contract_det_ite}, $f_p$ can be reformulated as
    \begin{equation*}
    	f_p=\frac{(\mathrm{tr}(\Lambda)+\lambda)^{n_x n_h / 2}}{\mathrm{tr}(\Lambda)^{n_x n_h / 2}( 1 + \lambda m)^{n_x / 2}}.
    \end{equation*}
    Then, we consider $f_q=f_p^{2/n_x}$. By differentiating $f_q$ with respect to $\lambda$, we obtain
    \begin{equation} \label{eq:dev_lambda}
    	\frac{\partial f_q}{\partial \lambda}=\frac{(\mathrm{tr}(\Lambda)\!+\!\lambda)^{n_h-1}(n_h -\mathrm{tr}(\Lambda)m+(n_h\!-\!1)m\lambda)}{\mathrm{tr}(\Lambda)^{n_h}(1+m\lambda)^2}.
    \end{equation}

    For the ``if'' part of the statement, there exists $u \in \mathcal{U}_o$ such that $n_h -\mathrm{tr}(\Lambda)m < 0$. Also note that $\lambda > 0$ and $n_h-1 > 0$. We can obtain $\frac{\partial f_q}{\partial \lambda} < 0$ when {$\lambda \in (0,\lambda^{\star})$}, where
    \begin{equation} \label{eq:opt_lambda}
    	\lambda^{\star}:=(\mathrm{tr}(\Lambda)m - n_h)/((n_h-1)m),
    \end{equation}
    based on \eqref{eq:dev_lambda}. Moreover, it can be verified that $f_q = 1$ when $\lambda = 0$. Note that $f_q < 1$ is equivalent to $f_p < 1$. Therefore, we can conclude that $f_p < 1$ for any $\lambda \in (0,\lambda^{\star}]$.

    The ``only if'' part of the statement can be demonstrated by contradiction. Suppose that there exist $\lambda\!>\!0$ and $u\!\in\!\mathcal{U}_o$ such that $f_p \!<\!1$ when $m \!\leq\! {n_h}/{\mathrm{tr}(\Lambda)},\forall u\in\mathcal{U}_o$. Then, we have $\frac{\partial f_q}{\partial \lambda} > 0,\forall \lambda>0, u\in\mathcal{U}_o$, which means $f_q$ is monotonically increasing for any $u\in\mathcal{U}_o$ with respect to $\lambda$ when $\lambda>0$. Note that $f_q=1$ when $\lambda = 0$. This result contradicts to the hypothesis that there exist $\lambda>0$ and $u\in\mathcal{U}_o$ such that $f_p <1$, which means there exists $u\in\mathcal{U}_o$ such that $m > {n_h}/{\mathrm{tr}(\Lambda)}$.
\end{proof}}

Recall that after the phase of formation, we have already obtained an ASE. Then, the idea of minimizing $\hat{\mu}(H_+,\Lambda_+)$ is equivalent to designing $u$ and $\lambda$ such that $f_p$ is minimized. It can be observed that $f_p$ is monotonically decreasing with respect to $m$. Besides, the proof of Theorem \ref{thm:ugen_contraction} asserts that $f_p$ is minimized by letting $\lambda=\lambda^{\star}$ when \eqref{eq:contract_criteria} is satisfied. Therefore, given current data matrices $H$, $\Lambda$, and state $x$, the active input design strategy is formulated as
\begin{equation}  \label{eq:contract_opt}
    u^{\star} = f_{\text{C}}(H,\Lambda,x):=\mathop{\arg\max}\limits_{u \in \mathcal{U}_o}
    \underbrace{\begin{bmatrix}
    		x \\[-.9ex] u
    	\end{bmatrix}^{\!\raisebox{-0.3ex}{$\scriptstyle\top$}}\!
    	(H \Lambda H^{\top})^{-1}
    	\begin{bmatrix}
    		x \\[-.9ex] u
    \end{bmatrix}}_{:=J_{\text{C}}(H,\Lambda,x,u)}
\end{equation}
to satisfy criterion \eqref{eq:contract_criteria} and minimize $f_p$. {Note that $H \Lambda H^{\!\top} \!\succ\! 0$. The problem in \eqref{eq:contract_opt} is a concave quadratic programming problem, which can be solved by referring to \cite{bronsteinApproximationConvexSets2008,yeAffineScalingAlgorithms1992,zwartGlobalMaximizationConvex1974}.} If \eqref{eq:contract_criteria} is not satisfied, a random input {$u \in \mathcal{U}_o$ will be applied, and the corresponding data sample will not be collected.} For clarity, the pseudo-code of the phase of formation is summarized in Algorithm \ref{alg:contraction}. Note that $T_L$ represents the expected learning time in the phase of contraction, and $\mathcal{C}(H^{F}, \dot{X}^{F},\Lambda^F)$ is the ASE obtained in the phase of formation.

{\begin{remark} \label{rem:HLambdaH}
	Note that the inequality \eqref{eq:contract_criteria} serves as a necessary and sufficient condition for ensuring the decrease of the upper bound. The inequality can be slightly modified to meet specific requirements on the rate of convergence. In particular, suppose that $\varepsilon\in(0,1)$, and $a^{\star}$ is the positive solution for the equation
	\begin{equation*}
		a^{1/n_h} + a^{1/n_h-1} = n_h(n_h-1)^{1-1/n_h}\varepsilon^{-2/(n_xn_h)}.
	\end{equation*}
	By replacing \eqref{eq:contract_criteria} with $m>(a^{\star}+1)/\mathrm{tr}(\Lambda)$, we have
	\begin{align*}
		&\mathrm{tr}(\Lambda)m(\mathrm{tr}(\Lambda)m-1)^{1/n_h-1} \\
		& = (\mathrm{tr}(\Lambda)m-1)^{1/n_h} + (\mathrm{tr}(\Lambda)m-1)^{1/n_h-1} \\
		& > n_h(n_h-1)^{1/n_h-1}\varepsilon^{-2/(n_xn_h)}.
	\end{align*}
	Then, we can obtain
	\begin{align*}
		\varepsilon^{2/(n_xn_h)} &> \frac{n_h(n_h-1)^{1/n_h-1}}{\mathrm{tr}(\Lambda)m(\mathrm{tr}(\Lambda)m-1)^{1/n_h-1}} \\
		& =\frac{\dfrac{n_h(\mathrm{tr}(\Lambda)m-1)}{m(n_h-1)}}{\mathrm{tr}(\Lambda)\Big(\dfrac{\mathrm{tr}(\Lambda)m-1}{n_h-1}\Big)^{1/n_h}} \\
		&  = \dfrac{\mathrm{tr}(\Lambda)+\lambda^{\star}}{\mathrm{tr}(\Lambda)( 1 + \lambda^{\star} m)^{1 / n_h}},
	\end{align*}
	which implies $f_p < \varepsilon$ according to Algorithm 2. Recall that $n_h$ data samples are collected in the phase of formation. We can conclude that $\hat{\mu}(H,\Lambda) < \varepsilon^{n-n_h} \hat{\mu}(H^F,\Lambda^F)$ in the phase of formation, where $n$ denotes the total number of collected data samples.
\end{remark}}

\color{black}\begin{algorithm}
	\caption{Phase of Contraction}
	\label{alg:contraction}
	\begin{algorithmic}[1]
		\STATE initialize $T_L, x, H\!\leftarrow\!H^{F}, \dot{X} \!\leftarrow \!\dot{X}^{F},\Lambda \!\leftarrow\! \Lambda^{F}, i \!\leftarrow\!0$
        \STATE \textbf{while} $i<T_L$
            \STATE \hspace{1em} $u \leftarrow f_{\text{C}}(H,\Lambda,x)$ and update $h$
            \STATE \hspace{1em} $m \leftarrow J_{\text{C}}(H,\Lambda,x,u)$
            \STATE \hspace{1em} \textbf{if} \eqref{eq:contract_criteria} holds
            \STATE \hspace{2em} compute $\lambda^{\star}$ by \eqref{eq:opt_lambda} and set $\Lambda \leftarrow \text{diag}(\Lambda,\lambda^{\star})$
            \STATE \hspace{2em} apply $u$ to \eqref{eq:sys_real} and measure the response $x$
            \STATE \hspace{2em} collect data sample $H \leftarrow [ H \ h]$, $ \dot{X} \leftarrow [\dot{X} \ x ]$
            \STATE \hspace{1em} \textbf{else}
            \STATE \hspace{2em} $u \leftarrow \text{a random element in } \mathcal{U}_o$
            \STATE \hspace{2em} apply $u$ to \eqref{eq:sys_real} and measure the response $x$
            \STATE \hspace{1em} \textbf{end if}
            \STATE \hspace{1em} set $i \leftarrow i +1$
            \STATE \textbf{end while}
            \STATE \textbf{return} $H, \dot{X}, \Lambda$
	\end{algorithmic}
\end{algorithm}

\section{Adaptive Data-driven Predictive Control}
\label{sec:Tube-based Data-driven Predictive Control}
{This section addresses Problem \ref{prob:adaptive_controller}, which is the design of the ATDPC. While many robust control frameworks, such as min-max predictive control, can handle uncertainties, we employ the adaptive tube-based predictive controller, which characterizes the uncertainty as a tube and updates it online and provides a natural solution for utilizing the learned ASE.

The schematic diagram of the ATDPC is presented in Fig. \ref{fig:schematic_ATDPC}. As illustrated, although the tube-based optimization problem is solved at each instant, the update of the dataset and the learning algorithm are only performed in an event-triggered fashion when an informative data sample is identified by the learning criterion.
}
\begin{figure}[tb]
	\setlength\belowcaptionskip{-3ex}
	\centerline{\includegraphics[width=\columnwidth]{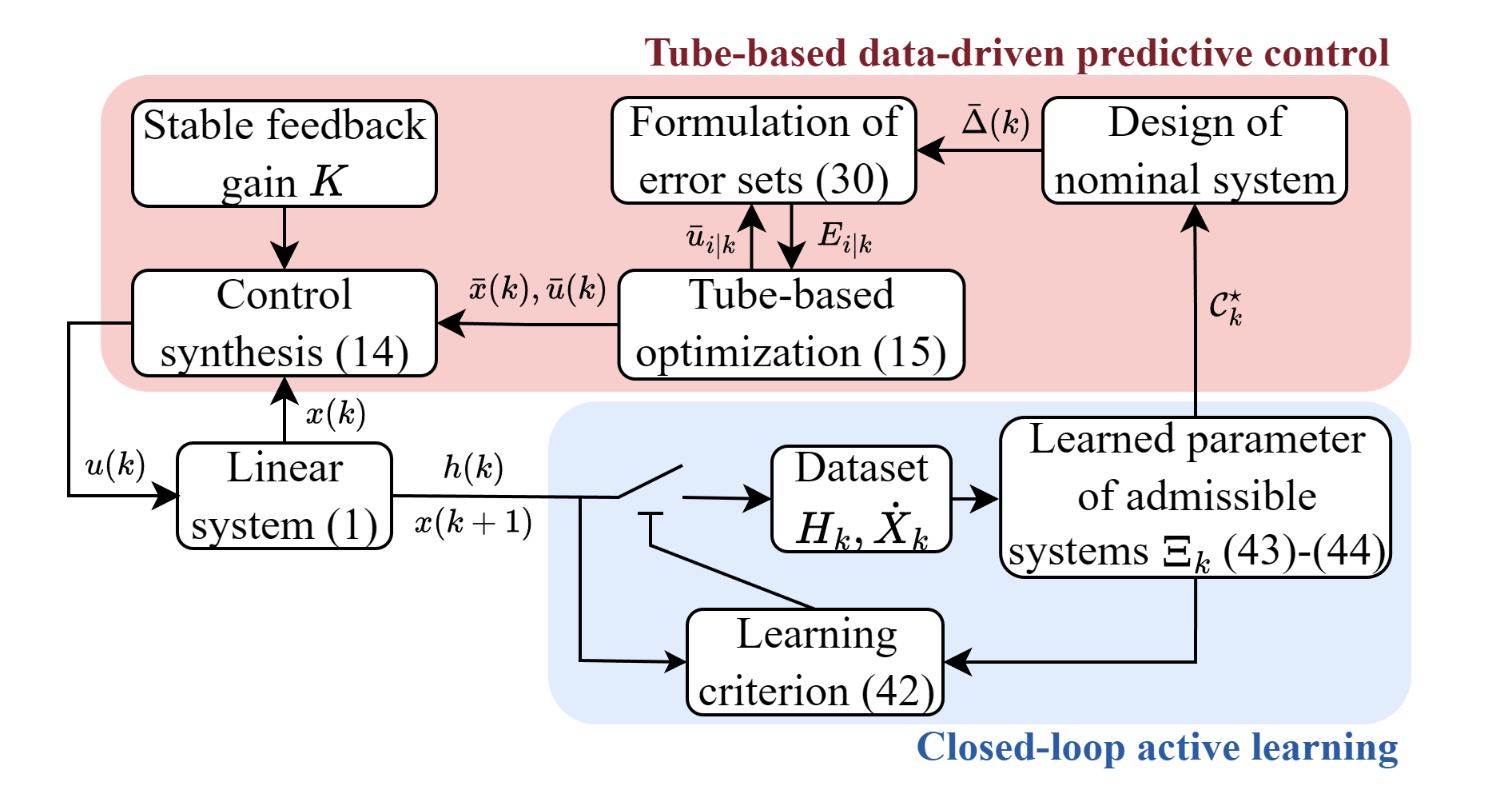}}
	\caption{Schematic diagram of the proposed ATDPC schemes.}
	\label{fig:schematic_ATDPC}
\end{figure}
\subsection{Design of Error Set and Terminal Constraint Set}
In this subsection, we introduce the design of the error set $E_{i|k}$ and analyze its boundedness and contraction properties. Then, we present the design of the terminal constraint set $\mathcal{X}_T$.

Based on \eqref{eq:sys_real}, \eqref{nominal_sys}, and \eqref{input}, we can derive
\begin{equation} \label{e_update}
	e(k+1) = \Delta_r
	\begin{bmatrix}
		I \\ K
	\end{bmatrix}
	e(k) + (\Delta_r - \bar{\Delta}(k))
	\begin{bmatrix}
		\bar{x}(k) \\ \bar{u}(k)
	\end{bmatrix} + w(k).
\end{equation}
The set $E_{i|k}$ that characterize all reachable errors can be analyzed given that $\Delta_r \in \mathcal{C}^{\star}_{k}$ and $w(t)\! \in \!\mathcal{W}$. Specifically, given $E_{0|k} = \{e(k)\}$, the set $E_{i|k}, i\geq 1$ can be computed by
{\begin{equation}
	\label{eq:e_outside_TCS}
	E_{i|k} :=  \left\{\Delta \begin{bmatrix}
		I \\ K
	\end{bmatrix} e + v \ \!\Big|\! \ \Delta \! \in \! \mathcal{C}_k^{\star}, e \! \in \! E_{i-1|k}, v \! \in \! V_k\right\}
\end{equation}
where
\begin{equation*}
	V_k\!:=\!\left\{(\Delta \!-\!\bar{\Delta}(k))\!\begin{bmatrix}
		\bar{x} \\ \bar{u}
	\end{bmatrix} \!+ \!w \ \!\Big|\! \ \Delta \! \in \! \mathcal{C}_k^{\star}, \bar{x} \! \in \! \mathcal{X}, \bar{u} \! \in \! \mathcal{U}, w \! \in \! \mathcal{W} \right\}.
\end{equation*}
}

Note that the solution of \eqref{eq:TDPC_OPT} relates closely to the error set and the performance of the ATDPC can only be guaranteed if the error is bounded. Before we further investigate the error set, a general property relating to the inclusion relationship between two matrix ellipsoids is first introduced.
\begin{lemma}
	\label{lem:E_inclusion}
	Let $\mathcal{E}$ and $\mathcal{E}'$ be matrix ellipsoids. Then the relationship $\left\{ Z-Z'_c \ | \ Z \in \mathcal{E}'\right\} \subseteq \left\{ Z-Z_c \ | \ Z \in \mathcal{E}\right\}$ holds if $\mathcal{E}' \subseteq \mathcal{E}$, where $Z_c$ and $Z'_c$ are the centers of $\mathcal{E}$ and $\mathcal{E}'$, respectively.
\end{lemma}
\begin{proof}
	{For simplicity, we denote $\left\{ Z-Z'_c \ | \ Z \in \mathcal{E}'\right\}$ by $\mathcal{S}'$ and $\left\{ Z-Z_c \ | \ Z \in \mathcal{E}\right\}$ by $\mathcal{S}$. Suppose that the conformable matrices ${M,N,M',N'} \succ 0$.} Then, without loss of generality, we assume
	\begin{equation*}
		\begin{split}
			\mathcal{E} &:= \left\{ {Z \ | \ (Z-Z_c)^{\top}M(Z-Z_c) \preceq N }\right\}, \\[-.3ex]
			\mathcal{E}' &:= \left\{ {Z \ | \ (Z-Z_c')^{\top}M'(Z-Z_c') \preceq N' }\right\}.
		\end{split}
	\end{equation*}
	Let $z := {Z-Z_c}$, $d := {Z_c'-Z_c}$, then standard computations reformulate $\mathcal{E},\mathcal{E}',\mathcal{S}$, and $\mathcal{S}'$ as
	\begin{equation*}
		\begin{split}
			\mathcal{E} &= \left\{ {z}+Z_c \ | \ z^{\top}{M}z \preceq {N}\right\}, \\[-.3ex]
			\mathcal{E}' &= \left\{ {z}+Z_c \ | \ (z-d)^{\top}{M'}(z-d) \preceq {N'}\right\}, \\[-.3ex]
			\mathcal{S} &= \left\{ z \ | \ z^{\top}{M}z \preceq {N}\right\}, \\[-.3ex]
			\mathcal{S}' &= \left\{ z \ | \ z^{\top}{M'}z \preceq {N'}\right\}.
		\end{split}
	\end{equation*}
	{Recall that $\mathcal{E}' \subseteq \mathcal{E}$. Then, we have
	\begin{equation*}
		\begin{bmatrix}
			I \\ z
		\end{bmatrix}^{\top}
		\begin{bmatrix}
			N & 0 \\ 0 & -M
		\end{bmatrix}
		\begin{bmatrix}
			I \\ z
		\end{bmatrix}
		\succeq 0
	\end{equation*}
	holds for all $z$ with
	\begin{equation}
		\label{eq:lem_Einclusion_proof_0.2}
		\begin{bmatrix}
			I \\ z
		\end{bmatrix}^{\top}
		\begin{bmatrix}
			N'- d^{\top}{M'}d &  d^{\top}{M'} \\ {M'}d & -M'
		\end{bmatrix}
		\begin{bmatrix}
			I \\ z
		\end{bmatrix}
		\succeq 0.
	\end{equation}
	Note that if we let $z=d$, the left hand of \eqref{eq:lem_Einclusion_proof_0.2} is equivalent to $N' \succ 0$, which verifies the Slater condition.} Then, Lemma \ref{lem:S-lemma} asserts that there exists a scalar $\alpha \geq 0$ such that
	\begin{equation}
		\label{eq:lem_Einclusion_proof_1}
		\begin{bmatrix}
			{N-\alpha N'}+ \alpha d^{\top}{M'}d & -\alpha d^{\top}{M'} \\
			- \alpha {M'}d & \alpha {M'} - {M}
		\end{bmatrix}
		\succeq 0.
	\end{equation}
	For compactness, we denote $M_\alpha:=\alpha {M'} - {M}$. By invoking \cite[Theorem 1.19]{zhang2006schur}, there exists a matrix ${L}$ such that
	\begin{equation}
		\label{eq:lem_Einclusion_proof_2}
		-\alpha d^{\top}{M'} = {L} M_\alpha.
	\end{equation}
	Furthermore, \cite[Theorem 1.20]{zhang2006schur} affirms that
	\begin{equation}
		\label{eq:lem_Einclusion_proof_3}
		{N-\alpha N'}+ \alpha d^{\top}{M'}d - \alpha d^{\top}{M'} M_\alpha^{\dagger} \alpha {M'}d
		\succeq 0
	\end{equation}
	where $\dagger$ denotes the generalized inverse of the matrix. Note that $\alpha \! \neq0$; otherwise the second diagonal block of \eqref{eq:lem_Einclusion_proof_1} equals to $- {M}$, which is a negative definite matrix and contradicts \eqref{eq:lem_Einclusion_proof_1}. {Note that \eqref{eq:lem_Einclusion_proof_3} is equivalent to
	\begin{equation}
		\label{eq:lem_Einclusion_proof_3.1}
		\begin{split}
			&{N-\alpha N'}+ \alpha d^{\top}{M'}(\alpha M')^{-1}(\alpha d^{\top}{M'})^{\top} \\
			&- \alpha d^{\top}{M'} M_\alpha^{\dagger} \alpha {M'}d \succeq 0.
		\end{split}
	\end{equation}
	After eliminating $\alpha d^{\top}{M'}$ via \eqref{eq:lem_Einclusion_proof_2}, \eqref{eq:lem_Einclusion_proof_3.1} becomes
	\begin{equation}
		\label{eq:lem_Einclusion_proof_3.2}
		\begin{split}
			&{N\!-\!\alpha N'}\!+\! {L} M_\alpha(\alpha M')^{-1}M_\alpha  {L}^{\!\!\top} \!\!-\! {L} M_\alpha M_\alpha^{\dagger}M_\alpha  {L}^{\!\!\top}\!\succeq\! 0.
		\end{split}
	\end{equation}
	Also note that $M_\alpha M_\alpha^{\dagger} M_\alpha = M_\alpha$ and $M_\alpha(\alpha M')^{-1}M_\alpha \!=M_\alpha - M + M(\alpha \!M')^{-1}M$. Then, \eqref{eq:lem_Einclusion_proof_3.2} is equivalent to
	\begin{equation}
		\label{eq:lem_Einclusion_proof_4}
		{N-\alpha N'} - {L(M-M(\alpha M')^{-1}M)L^{\top}}
		\succeq 0.
	\end{equation}}Note that we have $\alpha {M'} - {M} \succeq 0$ from \eqref{eq:lem_Einclusion_proof_1}. By standard Schur complement argument, ${M-M(\alpha M')^{-1}M}\!\succeq\! 0$ can be obtained, which implies ${N-\alpha N'} \!\succeq\! 0$. Now, we conclude that
	\begin{equation}
		\label{eq:lem_Einclusion_proof_5}
		\begin{bmatrix}
			{N-\alpha N'} & 0 \\
			0 & \alpha {M'} - {M}
		\end{bmatrix}
		\succeq 0,
	\end{equation}
	which proves the lemma by invoking Lemma \ref{lem:S-lemma}.
\end{proof}

{Now, we are in a position to demonstrate the boundedness and the shrinking properties of the reachable error sets given a feedback gain $K$ that stabilizes all systems in $\mathcal{C}^{\star}_{0}$ and a series of contracting ASEs. These properties are summarized in the following lemma.}

{\begin{lemma}
	\label{lem:error_set_inclusion}
	For the nominal system \eqref{nominal_sys}, if there exist a feedback gain $K$ and a matrix $P\succ 0$ such that
	\begin{equation} \label{eq:Einclusion_ass1}
		(\Delta [ I \ K^{\top}]^{\top})^{\top} P (\Delta [ I \ K^{\top}]^{\top}) - P \prec 0, \ \forall \Delta \in \mathcal{C}^{\star}_{0},
	\end{equation}
	and $e(0)$ is bounded, then the following statements hold.
	\begin{enumerate}
		\item there exists a bounded set $E_0$ such that $E_{i|0} \subseteq E_0$, $\forall i \in \mathbb{Z}_{\geq 0}$.
		\item it holds that $E_{i|k+1} \subseteq E_{i+1|k}$ if $\mathcal{C}_{k+1}^{\star} \subseteq \mathcal{C}_{k}^{\star}$, $\forall k \in \mathbb{Z}_{\geq 0}$, $i \in \mathbb{Z}_{\geq 0}$.
	\end{enumerate}
\end{lemma}
\begin{proof}
	For the first statement, we consider $\!E_{i|0}^0\!$ and $\!E_{i|0}^1$, where $E_{i|0}^0$ equals $E_{i|0}$ when $e(0)=0$ and $E_{i|0}^1$ characterizes the propagation of $e(0)$. {Let $E_{0|0}^0=\{0\}$ and $E_{0|0}^1=\{e(0)\}$.} For $i\in\mathbb{Z}_{\geq 1}$, $E_{i|0}^0$ and $E_{i|0}^1$ can be computed by
	\begin{equation}
		\begin{split}
			E_{i|0}^0 &=  \left\{\Delta \!
			\begin{bmatrix}
				I \\ K
			\end{bmatrix}e+v \ \!\big|\! \ \Delta  \in\mathcal{C}_0^{\star}, e  \in  E_{i-1|0}^0, v  \in  V_0\right\}, \\
			E_{i|0}^1 &=  \left\{\Delta \!
			\begin{bmatrix}
				I \\ K
			\end{bmatrix}e \ \!\big|\! \ \Delta  \in  \mathcal{C}_0^{\star}, e  \in  E_{i-1|0}^1\right\}.
		\end{split}
	\end{equation}
	It can be easily verified that {$E_{i|0} \subseteq E_{i|0}^0 \oplus E_{i|0}^1$ for all $i \in \mathbb{Z}_{\geq0}$}. Thus, the existence of $E_0$ can be guaranteed if there exist bounded sets $E_0^0$ and $E_0^1$ such that $E_{i|0}^0 \subseteq E_0^0,E_{i|0}^1 \subseteq E_0^1,\forall i \in \mathbb{Z}_{\geq 0}$. Given the assumption \eqref{eq:Einclusion_ass1} and the fact that $E_{0|0}^0 = \{0\}$, the existence of such $E_0^0$ follows directly from \cite[Lemma 5]{russoTubeBasedZonotopicDataDriven2023}. Additionally, by using \eqref{eq:Einclusion_ass1} and invoking \cite[Theorem 8]{huNonconservativeMatrixInequality2010}, there exist $\eta_1,\eta_2 > 0$ such that $\Vert e_i \Vert_2 \!<\! \eta_1e^{-\eta_2i}\Vert e(0) \Vert_2, \forall e_i\!\in\! E_{i|0}^0, i\in\mathbb{Z}_{\geq1}$. Note that since $e(0)$ is bounded, we can also conclude that there exists an $E_0^1$ such that $E_{i|0}^1 \!\subseteq\! E_0^1,\forall i \!\in\! \mathbb{Z}_{\geq 0}$.

	For the second statement, we will prove it by mathematical induction (MI). When $i = 0$, we have $E_{0|k+1} = \{e(k+1)\} = \{\Delta_r [ I \ K^{\top}]^{\top}\!e(k)+(\Delta_r - \bar{\Delta}(k))[\bar{x}(k)^{\top} \ \bar{u}(k)^{\top} ]^{\top} \!+ w(k)\},\forall k \in \mathbb{Z}_{\geq 0}$, where $\Delta_r \! \in \! \mathcal{C}_{k}^{\star},\bar{x}(k) \! \in \! \mathcal{X}, \bar{u}(k) \! \in \! \mathcal{U}, w(k) \! \in \! \mathcal{W}$. By the definition of $E_{i|k}$, we can easily conclude that $E_{0|k+1} \subseteq E_{1|k}$. Furthermore, assume that the statement holds for all $k \! \in \! \mathbb{Z}_{\geq 0}, i \! \in \! \mathbb{Z}_{[0,j-1]}$. Then, for $i=j$, we have
	\begin{equation*}
		\begin{split}
			E_{j|k+1} \! &= \! \left\{\Delta \!
			\begin{bmatrix}
				I \\ K
			\end{bmatrix}\!e\!+\!v \! \ \big| \! \ \Delta \! \in \! \mathcal{C}_{k+1}^{\star},e \! \in \! E_{j-1|k+1}, v \! \in \! V_{k+1}\right\},\! \\
			E_{j+1|k} \! &= \! \left\{\Delta \!
			\begin{bmatrix}
				I \\ K
			\end{bmatrix}\!e\!+\!v \! \ \big| \! \ \Delta \! \in \! \mathcal{C}_{k}^{\star}, e \! \in \! E_{j|k}, v \! \in \! V_k\right\}.
		\end{split}
	\end{equation*}
	Note that $\mathcal{C}_{k+1}^{\star} \subseteq \mathcal{C}_{k}^{\star}$ is guaranteed by \eqref{eq:ETL_strategy}. By letting $\mathcal{E}=\mathcal{C}_{k}^{\star}$, $\mathcal{E}'\!=\mathcal{C}_{k+1}^{\star}$, $Z_c=\bar{\Delta}(k)$, and $Z_c'\!=\bar{\Delta}(k+1)$, we have $V_{k+1} \subseteq V_k$ by invoking Lemma \ref{lem:E_inclusion}. Besides, $E_{j-1|k+1} \subseteq E_{j|k}$ can be obtained from the MI assumption. We can now conclude that the second statement holds for all $k \in \mathbb{Z}_{\geq 0}, i \in \mathbb{Z}_{\geq 0}$.
\end{proof}}

With the bounded reachable error set $E_0$ and the definition of the terminal constraint set $\mathcal{X}_T$, the parameter $L_T$ of the terminal constraint for any given matrices $P_T$ and $K$ can be obtained by solving an optimization problem:
{\begin{equation} \label{eq:Terminal_Set}
	\begin{split}
		&\max_{L_T} \ L_T\\[-1ex]
		& \ \text{s.t.} \ Kx \in \mathcal{U} \ominus KE_0, \ \forall x \in \{x \ | \ x^{\top}P_{T}x \leq L_T\}.
	\end{split}
\end{equation}}

By revisiting the formulation of $E_{i|k}$ in \eqref{eq:e_outside_TCS} and invoking Lemma \ref{lem:E_inclusion}, we can conclude that the size of $E_{i|k}$ correlates with $\mathcal{C}_{k}^{\star}$. Recall that a smaller $E_{i|k}$ will enable a less conservative solution for \eqref{eq:TDPC_OPT} and result in a more desirable control performance. Besides, the second statement of Lemma \ref{lem:error_set_inclusion} asserts that the contraction of the error set can be guaranteed if the ASE is contracting. This motivates the following subsection, in which we design a closed-loop active learning strategy to learn a contracting ASE by further exploring closed-loop data.

\subsection{{Closed-Loop Active Learning} for Contracting ASEs} \label{subsec:learn_Cs}
{When the controller is employed, the inputs are determined by \eqref{input} and cannot be designed freely. Therefore, an active learning strategy is proposed in this section to learn the contracting ASE by actively selecting closed-loop data samples.

Suppose that the dataset we obtained in the open-loop stage is $(H_{0},\dot{X}_{0},\Lambda_{0})$. We denote $\Xi(H_{0},\dot{X}_{0},\Lambda_{0})$ by $\Xi_0$ for brevity. Then, for $k \geq 1$, the active learning criterion is designed as}
\begin{equation}
	\label{eq:ETL_scheduler}
	\xi(h(k-1),x(k)) - \alpha \Xi_{k-1} \succeq 0.
\end{equation}
The corresponding learning strategy is presented as follows:
\begin{equation} \label{eq:ETL_update}
	\begin{split}
		H_k&=
		\left\{
		\begin{array}{lc}
			H_{k-1}, \ \text{if $\exists \alpha \geq 0$ such that \eqref{eq:ETL_scheduler} holds}\\[0pt]
			[H_{k-1}\ | \ h(k-1)],  \ \text{otherwise}\\
		\end{array}
		\right.
		\\
		\dot{X}_k&=
		\left\{
		\begin{array}{lc}
			\dot{X}_{k-1}, \ \text{if $\exists \alpha \geq 0$ such that \eqref{eq:ETL_scheduler} holds} \\[0pt]
			[\dot{X}_{k-1} \ | \ x(k)], \ \text{otherwise}\\
		\end{array}
		\right.
		\\
		\Xi_k&=
		\left\{
		\begin{array}{lc}
			\Xi_{k-1}, \ \text{if $\exists \alpha \geq 0$ such that \eqref{eq:ETL_scheduler} holds} \\[0pt]
			\Xi(H_k,\dot{X}_k,\Lambda_s), \ \text{otherwise}\\
		\end{array}
		\right.
	\end{split}
\end{equation}
where $\Xi_k$ determines the ASE learned by the closed-loop active learning strategy at instant $k$, namely,
	\begin{equation*}
		\mathcal{C}^{\star}_{k}:= \left\{ \Delta \ \Big| \ \begin{bmatrix}
			I \\ \Delta^{\top}
		\end{bmatrix}^{\top}\!
		\Xi_k
		\begin{bmatrix}
			I \\ \Delta^{\top}
		\end{bmatrix} \succeq 0 \right\}, k \in \mathbb{Z}_{\geq 0}.
\end{equation*}
{We then consider $\Xi_{k}^{\sigma}\!:=\!\Xi_{k}\!-\!\sigma\text{diag}(I_{n_x},0)$ with $\sigma\geq0$ and the corresponding ASE $\mathcal{C}_{k}^{\sigma}$, determined by $\Xi_{k}^{\sigma}$ in the same way.}

Based on the definition of matrix ellipsoids and Lemmas \ref{Lem:MoME} and \ref{lem:S-lemma}, $\mu(\mathcal{C}_{k-1}^{\sigma})$ decreases as $\sigma$ increases and $\mathcal{C}_{k-1}^{\sigma} \!\subseteq \mathcal{C}_{k-1}^{\star}$. Recall that we aim to learn a small ASE such that $\mathcal{C}^{\star}_{k} \!\subseteq\! \mathcal{C}^{\star}_{k-1}$. This objective can be achieved by learning for a $\mathcal{C}^{\star}_{k}$ that enables $\mathcal{C}^{\star}_{k} \subseteq \mathcal{C}_{k-1}^{\sigma}$ with a large $\sigma$. By invoking Lemma \ref{lem:S-lemma}, the closed-loop learning strategy is presented as follows:
\begin{equation}
	\label{eq:ETL_strategy}
	\begin{split}
		&\max_{\Lambda_s \in \mathcal{D}_{\text{col}(H_k)},\sigma\geq0} \sigma\\[-1ex]
		& \ \text{s.t. } \Xi_{k-1} - \Xi(H_k,\dot{X}_k,\Lambda_s) \succeq \sigma\text{diag}(I_{n_x}, 0).
	\end{split}
\end{equation}
The constraint $\sigma\! \geq\! 0$ ensures $\Xi_{k-1} \!- \Xi_{k} \!\succeq\! 0$, which implies $\mathcal{C}^{\star}_{k} \!\subseteq\! \mathcal{C}^{\star}_{k-1}$ by recalling the definition of $\mathcal{C}^{\star}_{k}$ and Remark \ref{rem:onedirect_slemma}.\footnote{CVX toolbox \cite{CVX} can be used to solve the semidefinite program \eqref{eq:ETL_strategy}.} Note that the execution of \eqref{eq:ETL_strategy} could be resource-consuming. The learning criterion \eqref{eq:ETL_scheduler} mitigates this issue by only triggering \eqref{eq:ETL_strategy} when an informative data sample is detected. {In fact, while problems \eqref{eq:ETL_scheduler} and \eqref{eq:ETL_strategy} have the same problem dimension, \eqref{eq:ETL_scheduler} involves only one scalar decision variable, which is $\text{col}(H_k)$ fewer than that of \eqref{eq:ETL_strategy}, where $\text{col}(H_k)$ denotes the size of the dataset.} An intuitive interpretation of \eqref{eq:ETL_scheduler} is that if all admissible systems in the current ASE are compatible with the latest data sample, the learning criterion will not be satisfied according to Lemma \ref{lem:S-lemma}. Therefore, it is unnecessary to trigger the learning strategy in \eqref{eq:ETL_strategy} because the latest data sample cannot exclude any admissible system in the ASE. {Additionally, to ensure a controllable computational load, we can set a limit on the number of collected data samples, and stop learning once the limit is reached.}

To further illustrate the rationale of the proposed active learning strategy, we investigate the influence of the uncollected data sample to the optimization problem \eqref{eq:ETL_strategy}. {Suppose that $\sigma^\star(k)$ is the optimal value of the objective function of  \eqref{eq:ETL_strategy}.} For convenience, we define $H_{k}^+ := [H_{k-1}\ h(k-1)]$ and $\dot{X}_{k}^+ := [\dot{X}_{k-1} \ x(k)]$, which represent the collected dataset without enforcing the learning criterion \eqref{eq:ETL_scheduler}. Suppose that when $H_k=H_{k}^+$ and $\dot{X}_k=\dot{X}_{k}^+$, the optimal value of the objective function of \eqref{eq:ETL_strategy} is $\sigma^+(k)$. {Recall that our goal is to demonstrate that the proposed criterion does not degrade overall learning performance by showing that the unselected data sample does not contribute to learning a contracting $\mathcal{C}_k^\star$. According to the definition of $\Xi_{k}^{\sigma}$ and \eqref{eq:opt_lambda}, if the unselected data sample benefits the learning process, we would have $\sigma^+(k) > \sigma^\star(k)$. Thus, the objective can be achieved by showing that $\sigma^\star(k) \geq \sigma^+(k)$. In the following theorem, we analyze the recursive feasibility of the learning strategy \eqref{eq:ETL_strategy} and prove that unselected samples are not beneficial for learning a contracting ASE.}
\begin{theorem} \label{thm:ETL}
	For ATDPC with learning strategy \eqref{eq:ETL_strategy} and the learning criterion \eqref{eq:ETL_scheduler}, the following statements hold:
	\begin{enumerate}
		\item the optimization problem \eqref{eq:ETL_strategy} is feasible for all $k\!\in\! \mathbb{Z}_{\geq 1}$.
		\item for $k \geq 1$, if there exist $\alpha \geq 0$ such that the learning criterion \eqref{eq:ETL_scheduler} holds, we have $\sigma^\star(k) \geq \sigma^+(k)$.
	\end{enumerate}
\end{theorem}
{
	\begin{proof}
		We prove the first statement by MI. Recall that $\Xi_0 = \Xi(H_0,\dot{X}_0,\Lambda_0)$. For $k=1$, if there exists $\alpha \geq 0$ such that \eqref{eq:ETL_scheduler} is satisfied, we have $H_1 = H_0$ and $\dot{X}_1 = \dot{X}_0$. It can be verified that $\Lambda_s = \Lambda_0$ and $\sigma = 0$ is a feasible solution to the optimization problem \eqref{eq:ETL_strategy}. {On the other hand}, if \eqref{eq:ETL_scheduler} is not satisfied, the constraint can still be satisfied by letting $\Lambda_s = \text{diag}(\Lambda_0,0)$ and $\sigma = 0$. For $k\geq 2$, we assume that there exists $\Lambda_{k-1}$ such that $\Xi_{k-2} \!-\! \Xi(H_{k-1},\dot{X}_{k-1},\Lambda_{k-1}) \!\succeq\! 0$. Then, the feasibility of \eqref{eq:ETL_strategy} can be analyzed following a similar line of argument to the case of $k=1$.

		For the second statement, we denote  $\text{col}(H_{k}^{+})$ by $n$ for compactness. Suppose that $\Lambda^+=\text{diag}(\lambda_1,\lambda_2,\dots,\lambda_n)$ is a solution corresponding to the optimal value $\sigma^+(k)$. We have
		\begin{equation} \label{eq:pf_ETL_1}
			\Xi_{k-1} \!-\! \sum\nolimits_{i=1}^{n}\!\!\lambda_i \xi((H_{k}^{+})_{[i]},({\dot{X}}_{k}^{+})_{[i]}) \!\succeq\! \sigma^+\!(k)\text{diag}(I_{n_x},0)
		\end{equation}
		according to \eqref{eq:ETL_strategy}. Besides, we also have
		\begin{equation} \label{eq:pf_ETL_2}
			\xi((H_{k}^{+})_{[n]},({\dot{X}}_{k}^{+})_{[n]}) - \alpha \Xi_{k-1} \succeq 0
		\end{equation}
		according to \eqref{eq:ETL_scheduler}. {Based on the assumption and the constraint in \eqref{eq:ETL_strategy}, we have $\alpha \geq 0$ and $\lambda_n \geq 0$.}

		{To prove the statement, we separately consider the case that $\alpha=0$ and $\alpha>0$. Specifically, the first case we consider is $\alpha = 0$ and the second case is $\alpha > 0$. We start by considering the case where $\alpha = 0$. By invoking \eqref{eq:pf_ETL_2}, we have $\xi((H_{k}^{+})_{[n]},({\dot{X}}_{k}^{+})_{[n]}) \succeq 0$, which implies
		\begin{equation}
			\Xi_{k-1} \!-\! \sum\nolimits_{i=1}^{n\!-\!1}\!\!\lambda_i \xi((H_{k}^{+})_{[i]},({\dot{X}}_{k}^{+})_{[i]}) \!\succeq\! \sigma^+\!(k)\text{diag}(I_{n_x}, 0)
		\end{equation}
	according to \eqref{eq:pf_ETL_1}. Noting that $(H_{k}^{+})_{[i]}\!=\!(H_{k})_{[i]}$ and $(\dot{X}_{k}^{+})_{[i]}\!=\!(\dot{X}_{k})_{[i]}$ when $i \in \mathbb{Z}_{[1,n-1]}$}, the constraint in \eqref{eq:ETL_strategy} is satisfied by letting $\Lambda_s=\text{diag}(\lambda_1,\lambda_2,\dots,\lambda_{n-1})$ and $\sigma=\sigma^+(k)$. Therefore, we have $\sigma^{\star}(k) \geq \sigma^{+}(k)$ when $\alpha = 0$.

		Next, for the case of $\alpha > 0$, by eliminating the term $\Xi_{k-1}$ using \eqref{eq:pf_ETL_1} and \eqref{eq:pf_ETL_2}, we can obtain
		\begin{equation} \label{eq:pf_ETL_3}
			\begin{split}
				&(1-\alpha\lambda_n)\xi((H_{k}^{+})_{[n]},({\dot{X}}_{k}^{+})_{[n]})\\
				&\succeq \!\alpha \!\sum\nolimits_{i=1}^{n\!-\!1}\!\!\lambda_i \xi((H_{k}^{+})_{[i]},({\dot{X}}_{k}^{+})_{[i]}) \!+\! \alpha\sigma^{+}\!(k)\text{diag}(I_{n_x}, 0).
			\end{split}
		\end{equation}
		{Note that the foremost right entry of \eqref{eq:pf_ETL_3} is positive semi-definite since $\alpha \!> \!0$ and $\sigma^{+}(k) \!\geq\! 0$.} Then, by pre- and postmultiplication of the right hand of \eqref{eq:pf_ETL_3} with $[I \ \Delta_r]$ and $[I \ \Delta_r]^{\top}$, it can be concluded that
		\begin{equation*}
			(1-\alpha \lambda_n )\begin{bmatrix}
				I \\ \Delta_{r}^{\top}
			\end{bmatrix}^{\top} \!
			\xi((H_{k}^{+})_{[n]},({\dot{X}}_{k}^{+})_{[n]})
			\begin{bmatrix}
				I \\ \Delta_{r}^{\top}
			\end{bmatrix} \succ 0.
		\end{equation*}
		according to \eqref{eq:ase_single}, which implies $1-\alpha \lambda_n > 0$. Since $\alpha > 0$ and $\lambda_n \geq 0$, we have $0 < 1-\alpha \lambda_n \leq 1$. Let $\tau:=1-\alpha \lambda_n$. Then, by eliminating the term $\xi((H_{k}^{+})_{[n]},({\dot{X}}_{k}^{+})_{[n]})$ using \eqref{eq:pf_ETL_1} and \eqref{eq:pf_ETL_2}, we obtain
		\begin{equation*}
				\tau\Xi_{k-1} - \sum\nolimits_{i=1}^{n-1}\lambda_i \xi((H_{k}^{+})_{[i]},({\dot{X}}_{k}^{+})_{[i]})\succeq\sigma^{+}(k)\text{diag}(I_{n_x}, 0),
		\end{equation*}
		which is equivalent to
		\begin{equation*}
			\Xi_{k-1} - \sum\nolimits_{i=1}^{n-1}\frac{\lambda_i}{\tau} \xi((H_{k}^{+})_{[i]},({\dot{X}}_{k}^{+})_{[i]})\succeq
			\frac{\sigma^{+}(k)}{\tau}\text{diag}(I_{n_x}, 0).
		\end{equation*}
		Note that $\text{diag}(\lambda_1,\lambda_2,\dots,\lambda_{n}) \in \mathcal{D}_n$, $\tau \in (0,1]$, and $\sigma^+ \geq 0$. It follows directly that the constraint in \eqref{eq:ETL_strategy} is satisfied when $\Lambda_s=\text{diag}(\lambda_1/\tau,\lambda_2/\tau,\dots,\lambda_{n-1}/\tau)$ and $\sigma=\sigma^{+}(k) / \tau$. Thus, we have $\sigma^{\star}(k) \geq \sigma^{+}(k) / \tau \geq \sigma^{+}(k)$.
\end{proof}

In conclusion, the learning criterion \eqref{eq:ETL_scheduler} provides a computationally efficient way to reduce the learning frequency while maintaining the learning performance.
}

\subsection{Recursive Feasibility and Stability of ATDPC}
\label{Recursive Feasibility and Stability}
In this section, we aim to show that the ATDPC is recursively feasible and could stabilize the true system with the proposed {closed-loop active learning} strategy. To this end, a lemma is first introduced to {discuss the selection of parameters $K$ and $P_T$ in the ATDPC using the learned ASE}.
\begin{lemma}
    \label{lem:recursive}
   For system \eqref{eq:sys_real} and its ASE $\mathcal{C}(H,\dot{X},\Lambda)$, assume that the generalized Slater condition \eqref{eq:SLemma_Slater} holds for $\Xi(H,\dot{X},\Lambda)$ and some $\bar{Z} \!\in\! \mathbb{R}^{n_h \times n_x}$. {Recall that $l(x,u)$ with $Q,R\succ 0$ is defined in \eqref{eq:stage_cost}}. Then, there exist matrices $K$ and $P_T$ such that
    \begin{equation} \label{eq:lem_recursive_1}
        V_f(Ax + Bu) - V_f(x) \leq -l(x,u)
    \end{equation}
    for all $x \in \mathcal{X}_T$ and ${\Delta} \in \mathcal{C}(H,\dot{X},\Lambda)$, where $u = K x$, if and only if there exist matrices $Y \in\mathbb{R}^{n_x\times n_x}$, $Y=Y^{\top} \succ 0$, $L \in\mathbb{R}^{n_u\times n_x}$, and scalar $\alpha \geq 0$ satisfying
    \begin{align}
        &\label{eq:csf_criteria_1}\begin{bmatrix}
            R^{-1} & L & 0\\
            L^{\top} & Y & Y \\
            0 & Y & Q^{-1}
        \end{bmatrix}
        \succ 0, \\
        &\label{eq:csf_criteria_2}\setlength{\arraycolsep}{2.5pt}
        \begin{bmatrix}
        	\begin{array}{ccc:ccc}
        		Y & 0 & 0 & 0 & 0 & 0 \\
        		0 & 0 & 0 & 0 & Y & 0 \\
        		0 & 0 & 0 & 0 & L & 0 \\ \hdashline
        		0 & 0 & 0 & \raisebox{0pt}[10pt]{$R^{-1}$}\! & L & 0 \\
        		0 & Y^{\!\top}\!\! & L^{\!\top}\! & L^{\!\top}\! & Y & Y \\
        		0 & 0 & 0 & 0 & Y & Q^{-1}\!
        	\end{array}
        \end{bmatrix}
        \!-\!\alpha
        \begin{bmatrix}
        	\Xi(H,\dot{X},\Lambda) & 0 \\
        	0 & 0
        \end{bmatrix}
        \!\succeq \!0.
    \end{align}
\end{lemma}
\vspace{1ex}
\begin{remark}
    \label{rem:recursive_feasibility_inXT}
    The linear matrix inequalities (LMIs) \eqref{eq:csf_criteria_1} and \eqref{eq:csf_criteria_2} in Lemma \ref{lem:recursive} provides the sufficient condition that there exists a controller $u=Kx$ such that \eqref{eq:lem_recursive_1} holds for $\forall x \in \mathcal{X}_T, \Delta \in \mathcal{C}$. Moreover, \eqref{eq:lem_recursive_1} also implies
    \begin{equation}
        \label{rem:RFinXT_proof_1}
        V_f(Ax + Bu) - V_f(x) \!\leq\! 0, \forall x \!\in \mathcal{X}_T, \Delta \!\in \mathcal{C}(H,\dot{X},\Lambda),
    \end{equation}
    from which we can conclude that there exists a controller $u=Kx$ such that $V_f(Ax + Bu) \leq L_T, \ \forall x \in \mathcal{X}_T, \Delta \in \mathcal{C}(H,\dot{X},\Lambda)$. As such, Lemma \ref{lem:recursive} provides a sufficient condition for the existence of the controller $u=Kx$ such that $Ax +Bu \in \mathcal{X}_T,\forall x \in \mathcal{X}_T, {\Delta} \in \mathcal{C}(H,\dot{X},\Lambda)$. {Moreover, a feasible solution exists if $\mathcal{C}(H,\dot{X},\Lambda)$ is sufficiently small and \eqref{eq:lem_recursive_1} is feasible for some matrices $K$, $P_T$, and the true systems parameters $(A,B)=(\mathcal{A},\mathcal{B})$, since a shrinking $\mathcal{C}(H,\dot{X},\Lambda)$ converges to the singleton that only includes the true system parameters.}
\end{remark}

Recall that in the {open-loop stage}, we obtained an ASE denoted by $\mathcal{C}^{\star}_{0}$. Then by considering the $\mathcal{C}^{\star}_{0}$ in the context of Lemma \ref{lem:recursive}, we can set matrices $P_T$ and $K$ in ATDPC as $P_T=Y^{-1}$ and $K=LP_T$ respectively, where $Y$ and $L$ are matrices such that \eqref{eq:csf_criteria_1} and \eqref{eq:csf_criteria_2} hold. Additionally, the conclusions drawn in Lemma \ref{lem:recursive} and Remark \ref{rem:recursive_feasibility_inXT} both hold for all systems in $\mathcal{C}^{\star}_{0}$ with $P_T$ and $K$, and the assumption \eqref{eq:Einclusion_ass1} in Lemma \ref{lem:E_inclusion} holds. Now, we are in a position to discuss the recursive feasibility and stazbility of ATDPC. By using Lemmas \ref{lem:error_set_inclusion}, \ref{lem:recursive} and Remark \ref{rem:recursive_feasibility_inXT}, the recursive feasibility and stability are analyzed in the following theorem.

\begin{theorem}
    \label{thm:ATDPC_stable}
    For system \eqref{eq:sys_real} and the ASE $\mathcal{C}^{\star}_0$, if there exist matrices $Y \in\mathbb{R}^{n_x\times n_x}$, $Y=Y^{\top} \succ 0$, $L \in\mathbb{R}^{n_u\times n_x}$, and a scalar $\alpha \geq 0$ that satisfy \eqref{eq:csf_criteria_1} and \eqref{eq:csf_criteria_2}, then the following statements hold:
    \begin{enumerate}
        \item  the optimization problem \eqref{eq:TDPC_OPT} is recursively feasible.
        \item the control sequence determined by \eqref{input} and \eqref{eq:TDPC_OPT} manipulates the state of the true system to a bounded set.
    \end{enumerate}
\end{theorem}
\begin{proof}
    Given the fact that there exist matrices $Y$ and $L$ such that \eqref{eq:csf_criteria_1} and \eqref{eq:csf_criteria_2} hold, we let $P=Y^{-1}$, $K=LP$. For recursive feasibility, we first assume that it is feasible at instant $k \in \mathbb{Z}_{\geq 0}$. Then there exists an optimal control sequence $\bar{u}_{\text{seq},k}^{\star}=(\bar{u}^{\star}_{0|k},\bar{u}^{\star}_{1|k},\dots,\bar{u}^{\star}_{N-1|k})$ such that $\bar{x}_{N|k}^{\star} \in \mathcal{X}_T$, $\bar{u}^{\star}_{i|k} \in \mathcal{U} \ominus K E_{i|k}$, $\forall i\in\mathbb{Z}_{\geq 0}$. At instant $k+1$, we consider a control sequence $\bar{u}_{\text{seq},k+1}$ where
    \begin{equation}
        \label{eq:feasible_solution}
        \bar{u}_{i|k+1}=
            \left\{
            \begin{array}{lc}
                \bar{u}^{\star}_{i+1|k}, \ i \in \mathbb{Z}_{[0,N-2]},\\[4pt]
                K \bar{x}_{N|k}^{\star}, \ i=N-1.
            \end{array}
            \right.
    \end{equation}
    {Recall that $\mathcal{C}_{k+1}^{\star} \subseteq \mathcal{C}_{k}^{\star}$ holds according to the learning algorithm \eqref{eq:ETL_strategy}, which implies $E_{i|k+1}\subseteq E_{i+1|k}$ by invoking the second statement of Lemma \ref{lem:error_set_inclusion}.} Then, note that $\bar{u}_{i|k+1} = \bar{u}^{\star}_{i+1|k} \in \mathcal{U} \ominus K E_{i+1|k},\forall i \in \mathbb{Z}_{[0,N-2]}$. {We can conclude that $\bar{u}_{i|k+1} \in \mathcal{U} \ominus K E_{i|k+1}, \forall i \in \mathbb{Z}_{[0,N-2]}$ since $E_{i|k+1}\subseteq E_{i+1|k}$.} Moreover, for $i=N-1$, since $\bar{x}_{N|k}^{\star}\in \mathcal{X}_T$, $\bar{u}_{N-1|k+1}$ is a feasible control input on the basis of \eqref{eq:Terminal_Set}. Besides, $\bar{x}_{N|k+1} \in \mathcal{X}_T$ can be derived by Remark \ref{rem:recursive_feasibility_inXT}. The recursive feasibility of \eqref{eq:TDPC_OPT} is thus proved.

    For the stability of the ATDPC, we reconsider the input sequence $\bar{u}_{\text{seq},k+1}$ which is defined in \eqref{eq:feasible_solution}. Suppose that the optimal solution is $\bar{u}_{\text{seq},k+1}^{\star}$. Then, we have
    {\begin{equation}
        \label{eq:thm_stab_proof_1}
        V_s(\bar{x}_{0|k+1},\bar{u}_{\text{seq},k+1}^{\star}) \leq V_s(\bar{x}_{0|k+1},\bar{u}_{\text{seq},k+1}).
    \end{equation}}Furthermore, we have
    {\begin{equation}
        \label{eq:thm_stab_proof_2}
        \begin{split}
            V_{\!s}(\bar{x}_{0|k+1},\!\bar{u}_{\text{seq},k+1}) \!=\! V_{\!s}(\bar{x}_{0|k},\!\bar{u}_{\text{seq},k}^{\star}) \!-\!l(\bar{x}_{0|k},\bar{u}_{0|k}^{\star})\!+\!\mathcal{F}
        \end{split}
    \end{equation}}where
    \begin{equation*}
        \mathcal{F} := V_f(\bar{\Delta}(k) \begin{bmatrix}
            I \\ K
        \end{bmatrix}) - V_f(\bar{x}_{N|k}^{\star}) + l(\bar{x}_{N|k}^{\star},K \bar{x}_{N|k}^{\star})
    \end{equation*}
    {By invoking Lemma \ref{lem:recursive}, \eqref{eq:lem_recursive_1} holds for all $x \in \mathcal{X}_T$ and ${\Delta} \in \mathcal{C}_0^\star$ by setting $u = K x$. Note that $\bar{\Delta}(k) \in \mathcal{C}_k^\star \subseteq \mathcal{C}_0^\star$ as $\mathcal{C}_{k+1}^\star \subseteq \mathcal{C}_k^\star$, and that $\bar{x}_{N|k}^\star \in \mathcal{X}_T$. Then, $\mathcal{F}<0$ follows from Lemma \ref{lem:recursive}.} Subsequently, according to \eqref{eq:thm_stab_proof_1} and \eqref{eq:thm_stab_proof_2}, we can obtain
    {\begin{equation*}
        V_s(\bar{x}_{0|k+1},\bar{u}_{\text{seq},k+1}^{\star}) - {V_s(\bar{x}_{0|k},\bar{u}_{\text{seq},k}^{\star})} \leq -l(\bar{x}_{0|k},\bar{u}_{0|k}^{\star}).
    \end{equation*}}With the definition of the stage cost function, we also have
    {\begin{equation*}
        V_s(\bar{x}_{0|k+1},\bar{u}_{\text{seq},k+1}^{\star}) - V_s(\bar{x}_{0|k},\bar{u}_{\text{seq},k}^{\star}) < 0, \ \forall \bar{x}_{0|k} \neq 0.
    \end{equation*}}Therefore, the nominal system is asymptotically stable for the origin. {Then, based on (1), (13), and (14), $e(k)$ propagates as
\begin{equation*}
e(k+1) = A_K
e(k) + (\Delta_r - \bar{\Delta}(k))
\begin{bmatrix}
\bar{x}(k) \\ \bar{u}(k)
\end{bmatrix} + w(k),
\end{equation*}
where $A_K := \Delta_r [I \ K^{\top}]^{\top}$ and $K$ is the feedback gain designed to stabilize the true system. Recall that $w(k) \in \mathcal{W}$, $\bar{x}(k) \in \mathcal{X}$, $\bar{u}(k) \in \mathcal{U}$, and $\bar{\Delta}(k) \in \mathcal{C}_k^{\star}\subseteq\mathcal{C}_0^{\star}$, where $\mathcal{W},\mathcal{X},\mathcal{U},\mathcal{C}_0^\star$ are bounded sets. Therefore, there exist $\delta_x,\delta_u,\delta_\Delta>0$ such that $\Vert \bar{x}(k) \Vert_2^2 \leq \delta_x$, $\Vert \bar{u}(k) \Vert_2^2 \leq \delta_u$, and $\Vert \Delta_r - \bar{\Delta}(k) \Vert_2^2 \leq \delta_\Delta$. Additionally, according to the definition of $\mathcal{W}$, we have $\Vert w(k) \Vert_2^2 \leq \delta$. Then, we can obtain that $$\Big\Vert \underbrace{(\Delta_r - \bar{\Delta}(k))
\begin{bmatrix}
\bar{x}(k) \\ \bar{u}(k)
\end{bmatrix} + w(k)}_{:=\rho(k)} \Big\Vert_2^2 \leq \delta_\Delta(\delta_x+\delta_u) + \delta:=\delta_s,$$
which implies
\begin{align*}
	\Vert e(N) \Vert_2^2 &= \Vert A_K^N e(0) + \sum\nolimits_{i=0}^{N-1} A_K^i \rho(N-1-i)\Vert_2^2 \\
	&\leq \Vert A_K^N e(0) \Vert_2^2 + \sum\nolimits_{i=0}^{N-1}\Vert A_K^i \rho(N-1-i) \Vert_2^2.
\end{align*}
Subsequently, recall that $A_K$ is Hurwitz. Suppose that an eigendecomposition of $A_K$ is $QEQ^{-1}$, and $\lambda_{\max} < 1$ is the spectral radius of $A_K$. It can be obtained that $\lim\limits_{N \rightarrow +\infty}\Vert A_K^N e(0) \Vert_2^2 = 0.$ Moreover, we have
\begin{align*}
&\lim\limits_{N \rightarrow +\infty}\sum\nolimits_{i=0}^{N-1}\Vert A_K^i \rho(N-1-i) \Vert_2^2 \\ &\leq \lim\limits_{N \rightarrow +\infty}\sum\nolimits_{i=0}^{N-1} \max_{\rho \in \{z | \Vert z \Vert_2^2 \leq \delta_s\}}\Vert A_K^i \rho \Vert_2^2
\end{align*}
as $\Vert \rho(k) \Vert_2^2 \leq \delta_s$. Since $A_K^i = QE^iQ^{-1}$, we have $$\max_{\rho \in \{z | \Vert z \Vert_2^2 \leq \delta_s\}}\Vert A_K^i \rho \Vert_2^2 \leq \delta_Q \delta_s \lambda_{\max}^i$$ where $\delta_Q := \Vert Q \Vert_2^2\Vert Q^{-1} \Vert_2^2$. Thus, we can obtain
\begin{align*}
&\lim\limits_{N \rightarrow +\infty}\sum\nolimits_{i=0}^{N-1}\Vert A_K^i \rho(N-1-i) \Vert_2^2 \\ &\leq \delta_Q \delta_s\lim\limits_{N \rightarrow +\infty}\sum\nolimits_{i=0}^{N-1} \lambda_{\max}^i \\ &= \delta_Q \delta_s(1-\lambda_{\max})^{-1}.
\end{align*}
Now, we can conclude that
\begin{align*}
&\lim\limits_{N \rightarrow +\infty} \Vert e(N) \Vert_2^2 \\ &\leq \lim\limits_{N \rightarrow +\infty}\Vert A_K^N e(0) \Vert_2^2 + \!\!\lim\limits_{N \rightarrow +\infty}\sum\nolimits_{i=0}^{N-1}\Vert A_K^i \rho(N\!-\!1\!-\!i) \Vert_2^2\\ & = \delta_Q \delta_s(1-\lambda_{\max})^{-1}.
\end{align*}
Recall that $x(k)=\bar{x}(k)+e(k)$. Therefore, the true system state will converges to a bounded set, which is bounded by $\{z | \Vert z \Vert_2^2 \leq \delta_Q \delta_s(1-\lambda_{\max})^{-1}\}$.}
\end{proof}

\begin{remark}
	By considering the set-valued ASE, we ensure that the ATDPC stabilizes the true system. {Note that the size of $E_{i|k}$ directly affects the solution space of the control optimization problem \eqref{eq:TDPC_OPT}, and that $E_{i|k}$ is positively correlated with the size of $\mathcal{C}_k^{\star}$. By incorporating the closed-loop active learning strategy, a series of contracting $\mathcal{C}_k^{\star}$ is obtained, which enlarges the solution space of \eqref{eq:TDPC_OPT}, thereby enabling less conservative control inputs and reducing control cost $V_s$.}
\end{remark}

\section{Numerical Examples}
\label{sec:Numerical Examples}
In this section, we illustrate the advantage and effectiveness of our results by numerical simulations.

\subsection{Visualization of {Open-Loop Active Learning}}
To visualize the efficiency of the proposed active learning strategy, we consider a scalar system \eqref{eq:sys_real} with $\mathcal{A}=1$, $\mathcal{B}=1$, and $\delta=1$. Assume that the input constraint for {open-loop active learning} is $\mathcal{U}_{o}=\{u \mid \Vert u \Vert_2 \leq 5\}$. As a comparison, we apply the input design strategies in \cite{2022_Beyond} and \cite{muller2023inputdesign}, which are denoted by ID-PE and ID-$\alpha$PE, respectively. The ID-PE selects inputs online to ensure the data matrix achieves full rank properties, while the ID-$\alpha$PE designs an offline input sequence that provides a lower bound on the minimum singular value of the data matrix.

\begin{figure}[t!]
	\centerline{\includegraphics[width=.9\columnwidth]{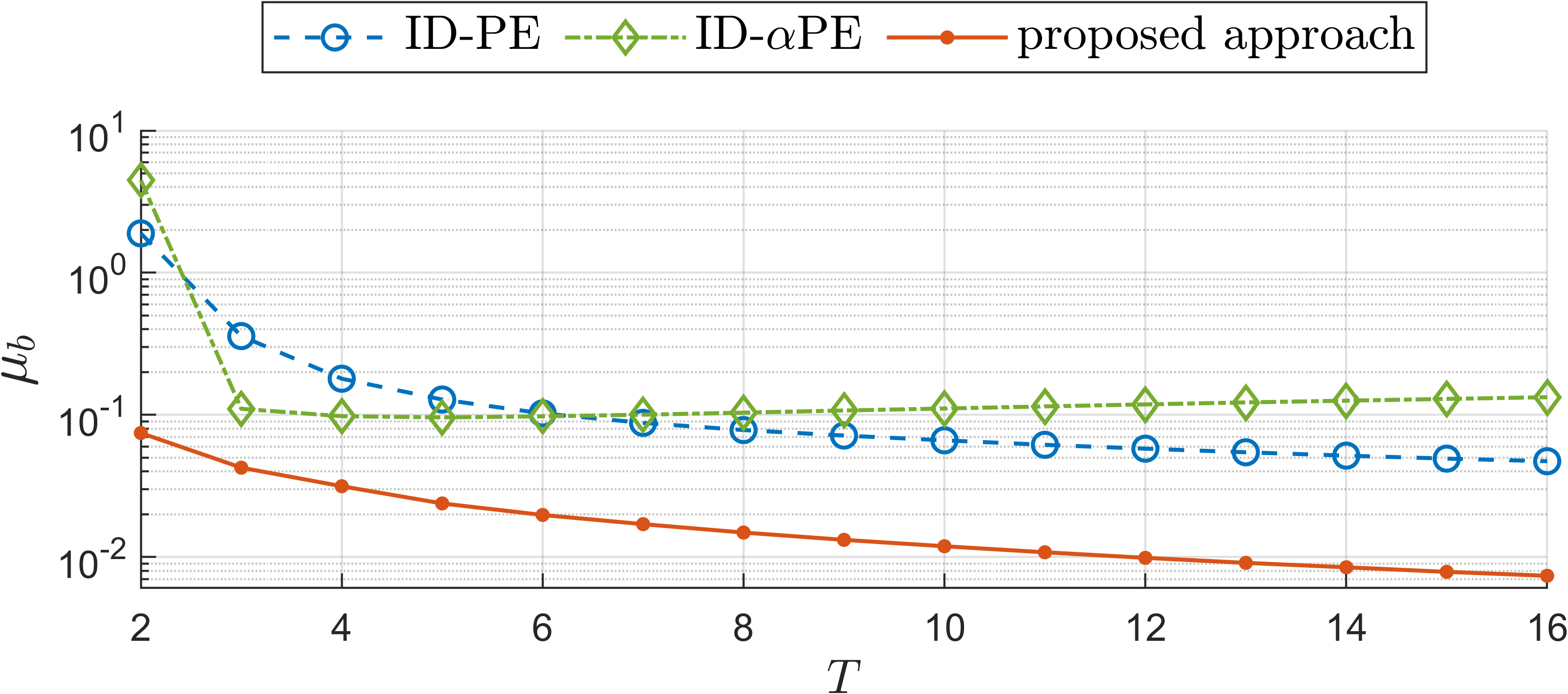}}
	\caption{{Volume of the ASE obtained by different approaches.}}
	\label{fig_modeling_vol}
	\vspace{4ex}
	\centerline{\includegraphics[width=.9\columnwidth]{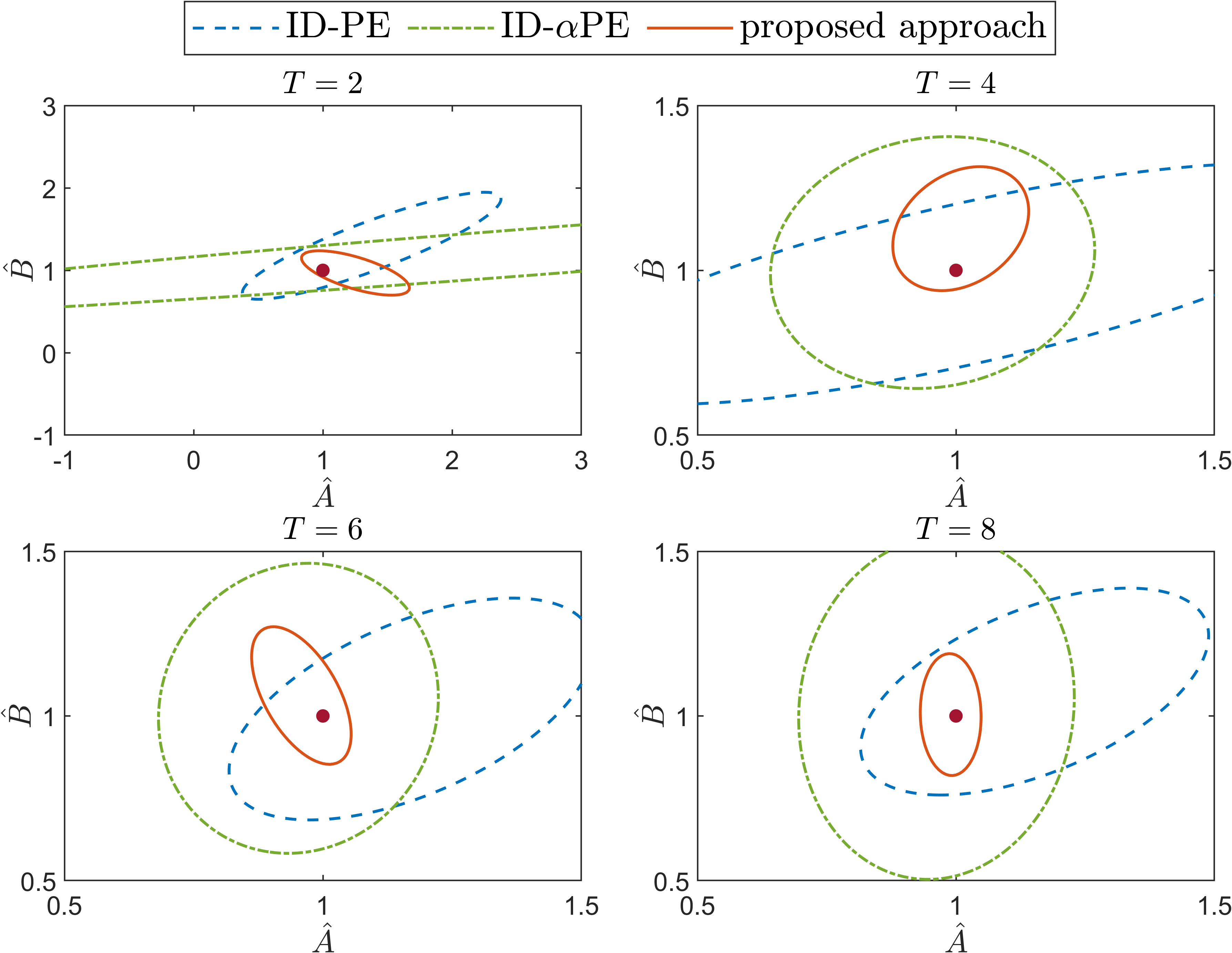}}
	\caption{Space of the ASE obtained by different approaches. The dark brown dot represents the true system.}
	\label{fig_modeling_space}
\end{figure}

To better represent the volume of the learned ASE, we define $\mu_b:=\det(F_e^{\top}E_e^{-1}F_e - G_e)^{{n_h}/{2}} \det(E_e^{-1})^{n_x/2}$. According to Lemma \ref{Lem:MoME}, $\mu_b = \mu(\mathcal{C}) / \beta$, where $\mu(\mathcal{C})$ is the exact volume of the ASE, and $\beta$ is a constant multiplier related only to $n_x$ and $n_u$. Fig. \ref{fig_modeling_vol} illustrates the average value of $\mu_b$ of the ASE obtained by ID-PE and ID-$\alpha$PE compared to those obtained by the proposed approach in a batch of 100 simulations, as the number of data samples increases from $T\!=\!2$ to $T\!=\!16$. Furthermore, we set $Q=1$ and $R=0.1$. In an instance of the simulation, the matrices obtained according to Lemma \ref{lem:recursive} are $Y = 0.2669$ and $L = -0.2676$. Additionally, we visualize the space of the learned ASEs in Fig. \ref{fig_modeling_space}, where the performance of the proposed approach can be observed more intuitively. {The proposed approach achieves a smaller average volume and is capable of learning a smaller ASE with fewer data samples, thereby verifying the effectiveness of the theoretical results in Lemma \ref{lem:lambda_square} and Theorem \ref{thm:ugen_contraction}, where the volume is considered as an informativity metric in the design of inputs.}

\subsection{Illustration of Active Learning and Control}
{In this section, we delve into the consideration of a more complex example to represent the learning and control performance of ATDPC. We consider system \eqref{eq:sys_real} with
\begin{equation*}
    \mathcal{A} \!=\!
    \begin{bmatrix}
        0.850 & -0.038 & -0.038 \\ 0.735 & 0.815 & 1.594 \\ -0.664 & 0.697 & -0.064
    \end{bmatrix}\!,\
     \mathcal{B} \!=\!
    \begin{bmatrix}
       1.431 & 0.705 \\ 1.62 & -1.129 \\ 0.913 & 0.369
    \end{bmatrix}\!,
\end{equation*}
and we let $Q = I_3$, $R = I_2 / 10$, and $\delta = 0.16$. The constraint for {open-loop active learning} is $\mathcal{U}_{o}\!=\!\{u \in \mathbb{R}^2 \mid \Vert u \Vert_2 \leq 2\}$. We first analyze the influence of learning strategies and number of data samples to the feasibility of ATDPC. For each number of collected data samples $T\in[5,24]$, we solve a batch of $100$ feasibility problem of ATDPC, count the number of feasible instances by $n_f$, and display ${n_f}/{100}$ for each approach in Fig. \ref{fig_feas_random}. The proposed approach allows the design of ATDPC {in more than $80$ instances out of all $100$ simulations} ($n_f \geq 0.8$) when $T \geq 5$ and ensures the feasibility ($n_f = 1$) when $T \geq 7$. In contrast, ID-PE enables $n_f \geq 0.8$ when $T \geq 19$, and ID-$\alpha$PE enables $n_f \geq 0.5$ when $T\!=\!10$. This observation is due to the fact that the proposed approach designs inputs online to minimize the volume of ASE. This implies less uncertainty about the system and thereby increases the possibility of enabling ATDPC. The active learning strategy shows its appealing feature of being able to effectively design controllers within limited data samples.}
\begin{figure}[t!]
	\centerline{\includegraphics[width=.9\columnwidth]{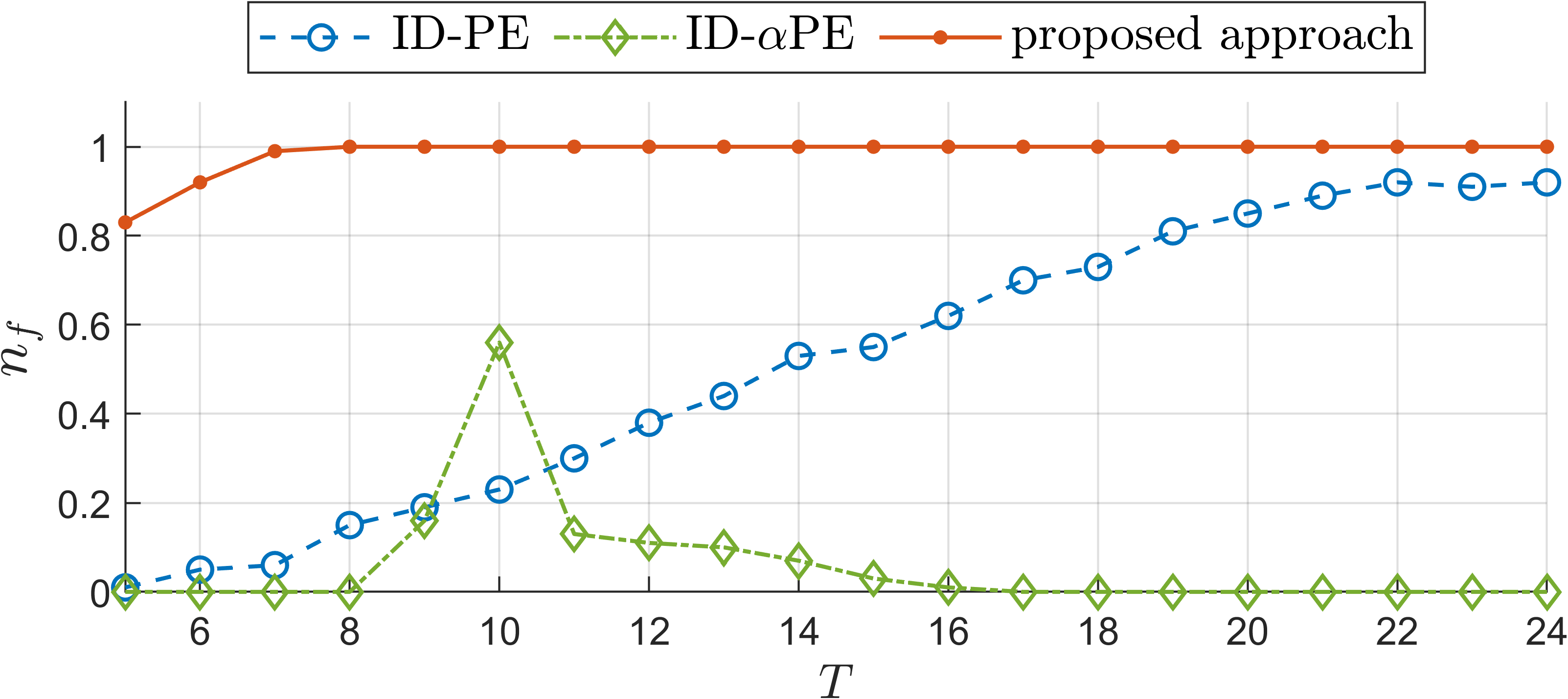}}
	\caption{Ratio of feasible problems using different approaches.}
	\label{fig_feas_random}
\end{figure}

{Furthermore, to illustrate the effectiveness of the proposed active learning strategy in {closed-loop} stage, we compared the data-driven predictive control designed utilizing ID-PE and ID-$\alpha$PE with the proposed ATDPC. We consider a generic reference tracking problem to show the advantages of the proposed active learning approach. The tracking duration is 100, and the reference $x_f(k)$ equals to $[-1.2742 , -5.1937 , -2.7653]^{\!\top}\!$ when $k\in[13,24] \cup [61,72]$, $[5.9144 , 5.1550 , 0.1472]^{\!\top}\!$ when $k\in[37,48]$, and $[0,0,0]^{\!\top}\!$ for other unspecified times instants. Besides, the initial point is set as $x_0 \!=\! [2,1,-1]^{\top}$, the {closed-loop} state and input constraint sets are $\mathcal{X}\!=\!\{x \!\in\! \mathbb{R}^3 \mid \Vert x \Vert_2 \!\leq\! 8\}$ and $\mathcal{U}\!=\!\{u \!\in\! \mathbb{R}^2 \mid \Vert u \Vert_2 \!\leq\! 3\}$, respectively. According to Lemma \ref{lem:recursive}, the matrices $Y$ and $L$ corresponding to the designed controller are
\begin{align*}
	Y\!\!=\!\!\begin{bmatrix}
		0.228   & 0.199  & -0.017\\
		0.199 &   0.454&   -0.12\\
		-0.017  & -0.12  &  0.316
	\end{bmatrix}\!,
	L \!=\! \!\begin{bmatrix}
		   -0.117 & -0.242 \\ -0.383 & 0.484 \\ -0.354 & 0.795
	\end{bmatrix}\!.
\end{align*}
\begin{figure}[t!]
	\centerline{\includegraphics[width=.9\columnwidth]{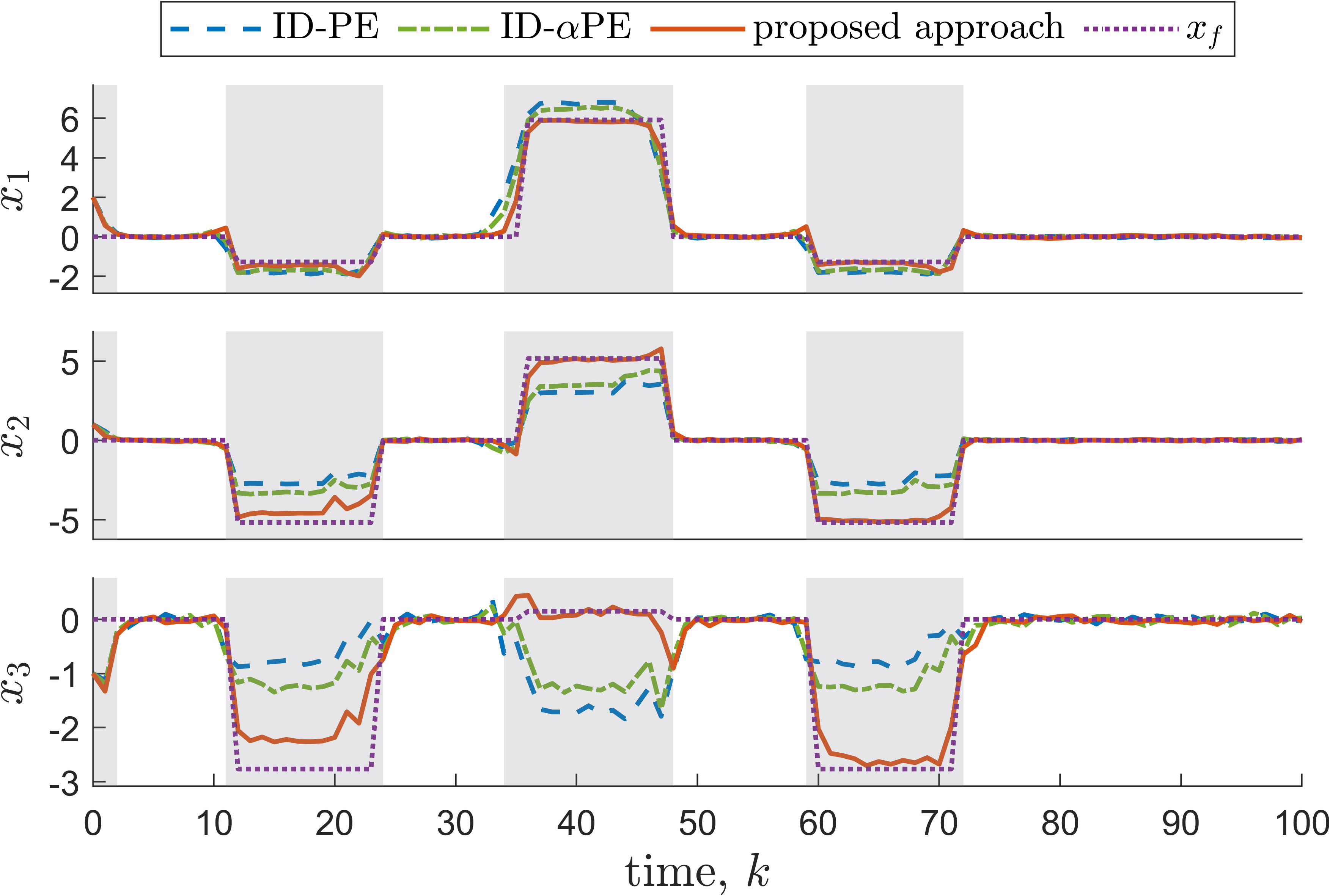}}
	\caption{{Closed-loop state trajectories of different controllers and the tracking reference.}}
	\label{AvR_TDPC}
	\vspace{4ex}
	\centerline{\includegraphics[width=.9\columnwidth]{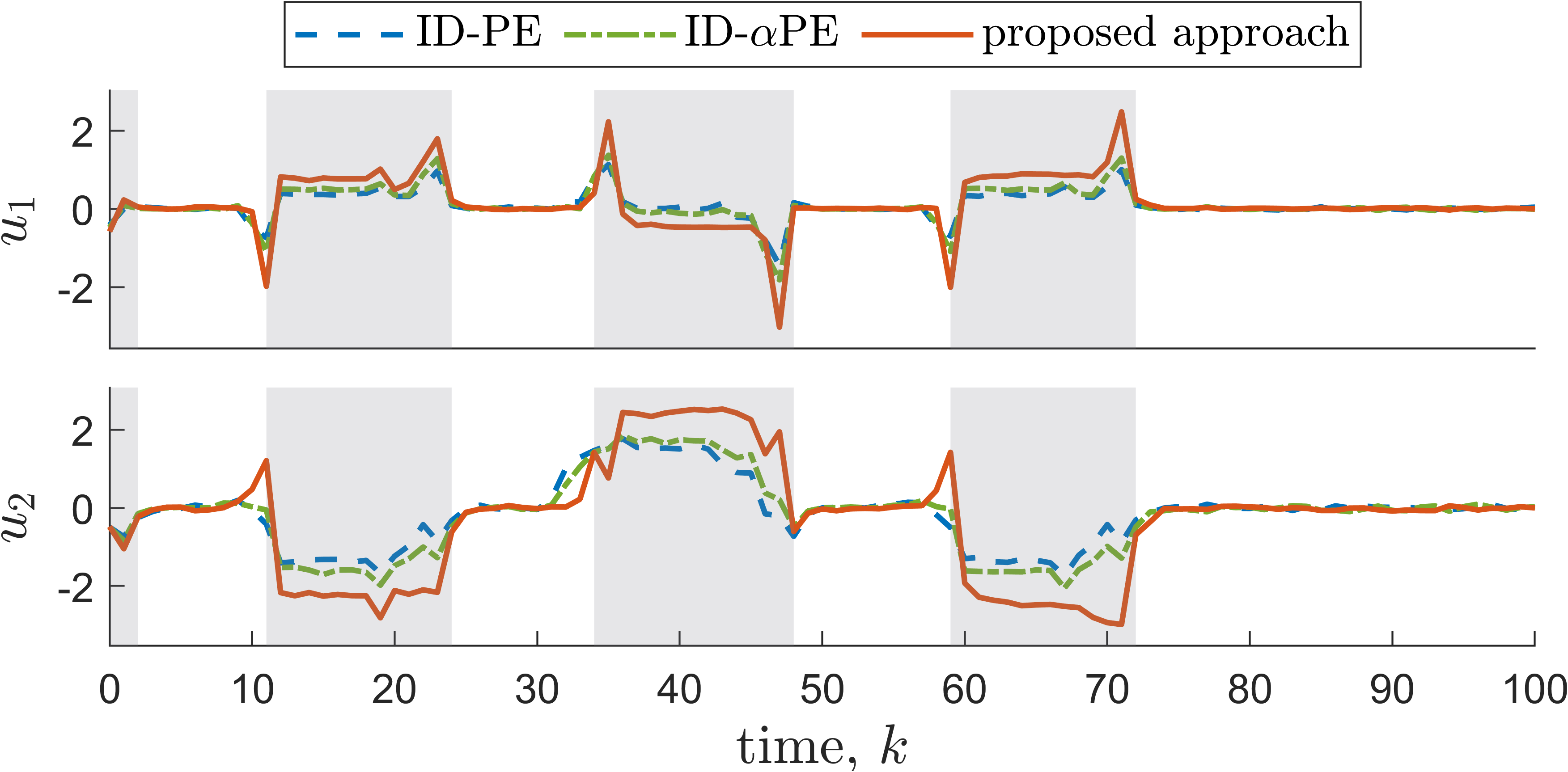}}
	\caption{{Closed-loop input sequence of different controllers.}}
	\label{AvR_TDPC_input}
\end{figure}
When $10$ data samples are collected using each approach, Fig. \ref{AvR_TDPC} shows the trajectories of system state and Fig. \ref{AvR_TDPC_input} shows the control input of different approaches. The gray parts of Fig. \ref{AvR_TDPC} and Fig. \ref{AvR_TDPC_input} represent the time instances when the learning criterion \eqref{eq:ETL_scheduler} is triggered. Note that when the system is stabilized near the origin ($k \geq 73$), the learning algorithm will not be triggered, and no further data samples will be collected. A limit on the number of data samples can be set in practice to prevent the collection of excessive samples.}}

It can be observed that ATDPC provides a less conservative input sequence which enables a better tracking performance. Moreover, although ATDPC cannot track $x_f$ when $k\in[13,24]$, it successfully tracks the same reference $x_f$ when $k\in[61,72]$. To further analyze the control performance of both methods, we introduce the performance index as follows $J_t \!=\! \sum_{k=0}^{100}(x(k)\!-\!x_f(k))^{\top}Q(x(k)\!-\!x_f(k))+u(k)^{\top}Ru(k)$. The performance index for each approach is represented in Table \ref{tab:control_cost}, which implies the improved control performance of ATDPC. {The improved performance is attributed to the ability of the proposed active learning approach to selectively learn from the closed-loop data, thereby verifying the effectiveness of the theoretical results in Theorems \ref{thm:ETL} and \ref{thm:ATDPC_stable}.}

\begin{table}[htbp]
	\centering
	\caption{{Comparison of the tracking costs.}}
	\resizebox{0.4\textwidth}{!}{ 
		\begin{tabular}{*{4}{c}} 
			\toprule
			  & ID-PE & ID-$\alpha$PE & proposed approach \\
			\midrule
			$J_t$ & $145.0326$ & $116.6529$ & $51.3070$ \\
			\bottomrule
		\end{tabular}
		\label{tab:control_cost}
}
\end{table}

\subsection{Robustness and Efficiency of Active Learning Strategies}
\begin{figure}[t!]
	\centerline{\includegraphics[width=.9\columnwidth]{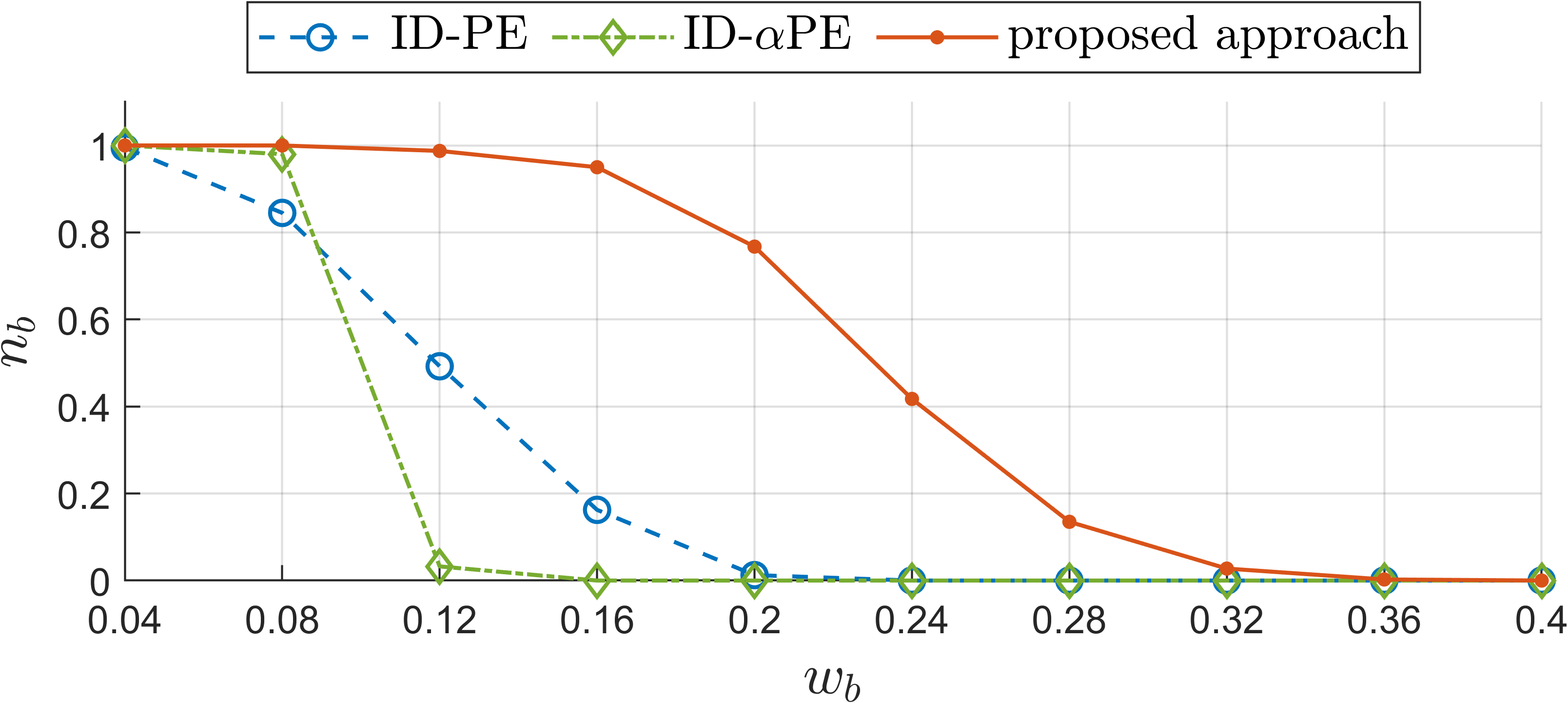}}
	\caption{Ratio of feasible problems under different noise levels.}
	\label{fig:feas_noise_level}
\end{figure}
{In this section, we further delve into the linearized and discretized longitudinal flight control system of Boeing 747 \cite{ishiharaDesignDiscretetimeIntegral1992a}, which is in the form of \eqref{eq:sys_real} with
\begin{equation*}
	\mathcal{A} \!=\!
	\begin{bmatrix}
		0.99 & 0.03 & -0.02 & -0.03 \\ 0.01 & 0.47 & 4.7 & 0 \\ 0.02 & -0.06 & 0.4 & 0 \\ 0.01 & -0.04 & 0.72 & 0.99
	\end{bmatrix}\!\!,
	\mathcal{B} \!=\!
	\begin{bmatrix}
		0.01 & 0.99 \\ -3.44 & 1.66 \\ -0.83 & 0.44 \\ -0.47 & 0.25
	\end{bmatrix}\!\!,
\end{equation*}
$Q \!=\! I_4$, $R \!=\! I_2$, $\mathcal{W}\!=\!\{w \in \mathbb{R}^4 \mid \Vert w \Vert_2 \leq w_b^2\}$, and $\mathcal{U}_{o}\!=\!\{u \in \mathbb{R}^2 \mid \Vert u \Vert_2 \leq 2\}$. To show the efficiency of the proposed approach, we first analyze its performance under different noise levels. For each $w_b$, we set and number of collected data samples $T=16$, solve a batch of 400 feasibility problems of ATDPC using \eqref{eq:csf_criteria_1} and \eqref{eq:csf_criteria_2} with the initial system state subjects to Gaussian distribution, count the number of feasible instances by $n_b$, and display the ratio of feasible instances $n_b / 400$ in Fig. \ref{fig:feas_noise_level}. It can be observed that the proposed active learning approach outperforms the comparative approaches in different noise levels and is able to handle a higher noise level.

Furthermore, to show the efficiency of the proposed approach, the computation times to solve \eqref{eq:ETL_scheduler} and \eqref{eq:ETL_strategy} with different number of collected data samples $T$ is presented in Fig. \ref{fig:computation_time}. The learning criterion \eqref{eq:ETL_scheduler} has the appealing feature that the associated computation time barely changes with the number of data samples whereas the time cost of solving \eqref{eq:ETL_strategy} grows linearly with the number of data samples.
\begin{figure}[t!]
	\centerline{\includegraphics[width=.9\columnwidth]{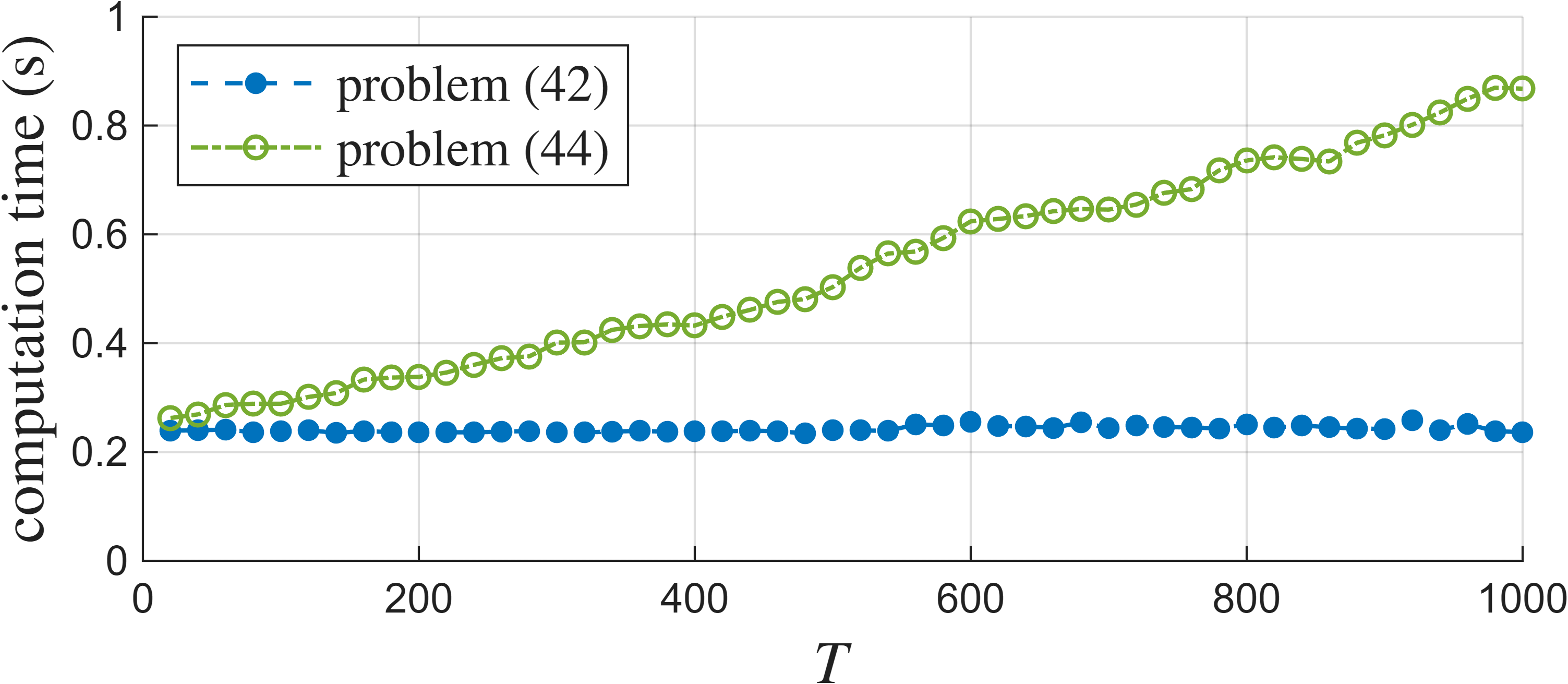}}
	\caption{Computation time comparison between solving \eqref{eq:ETL_scheduler} and \eqref{eq:ETL_strategy}.}
	\label{fig:computation_time}
\end{figure}

}

\section{Conclusion}
\label{sec:conclusion}
In this work, we investigated efficient data acquisition strategies for generating and selecting informative datasets for linear systems with bounded disturbances. {In the open-loop stage, we proposed an active input design strategy to obtain a small admissible system set. In the closed-loop stage, we introduced a learning criterion to selectively collect data samples that contribute to further reducing the volume of admissible system set. We have followed up by applying our active learning strategies to predictive control. Future research includes extending the proposed strategies to systems with stochastic disturbances and to distributed settings.}

\section*{Appendix}
\subsection{Proof of Lemma 3}
\begin{proof}
	Based on the definition of $\Xi(H,\dot{X},\Lambda)$, inequality \eqref{eq:admissible_system} is equivalent to
	\begin{equation} \label{eq:ase_ellip}
		\Delta E_e \Delta^{\top} + \Delta F_e + (\Delta F_e)^{\top} + {G_e} \preceq 0
	\end{equation}
	where
	\begin{equation} \label{eq:def_ellip_matrix}
		E_e \!:=\! H\Lambda H^{\top}\!, F_e \!:=\! -H\Lambda \dot{X}^{\top}\!, G_e \!:=\! \dot{X}\Lambda \dot{X}^{\top} \!- \mathrm{tr}(\Lambda) \delta I.
	\end{equation}
	The inequality in \eqref{eq:ase_ellip} is in the same form with \eqref{eq:matrix_ellipsoid_1}. Then $\mathcal{C}(H,\dot{X},\Lambda)$ is a matrix ellipsoid if $E_e \succ 0$ and $F_e^{\top}E_e^{-1}F_e - G_e \succ 0$. Note that $E_e \succ 0$ follows directly from the condition $H \Lambda H^{\top} \succ 0$. Besides, we can obtain
	\begin{equation} \label{eq:elli_cal}
		\begin{split}
			F_e^{\!\top}\!E_e^{-1}\!F_e \!-\! G_e\!&=\!\dot{X}\Lambda H^{\!\top}\!\!E_e^{-1}H \Lambda \dot{X}^{\!\top}\!\! +\! \mathrm{tr}(\Lambda) \delta I \!-\! \dot{X}\Lambda \dot{X}^{\!\top}\!\\
			&=\!\dot{X}(\Lambda H^{\!\top}\!E_e^{-1}\!H \Lambda\!-\!\Lambda) \dot{X}^{\!\top} \!\!+ \!\mathrm{tr}(\Lambda) \delta I.
		\end{split}
	\end{equation}
	By invoking \eqref{eq:matrix_sys}, we have
	\begin{equation} \label{eq:elli_cal_1}
		\begin{split}
			&\dot{X}(\Lambda H^{\top}E_e^{-1}H \Lambda-\Lambda) \dot{X}^{\top} \\
			&=(\Delta_r H+W)(\Lambda H^{\top}E_e^{-1}H \Lambda-\Lambda)(\Delta_r H+W)^{\top}.
		\end{split}
	\end{equation}
	Since $E_e \succ 0$, it holds that
	\begin{equation*}
		{H(\Lambda H^{\top}E_e^{-1}H \Lambda - \Lambda) = E_eE_e^{-1}H \Lambda -H \Lambda = 0.}
	\end{equation*}
	We can conclude that \eqref{eq:elli_cal_1} equals $W(\Lambda H^{\top}E_e^{-1}H \Lambda-\Lambda)W^{\top}\!$, and \eqref{eq:elli_cal} is equal to $W\Lambda H^{\top}E_e^{-1}H \Lambda W^{\top} \!+ \mathrm{tr}(\Lambda) \delta I -W\Lambda W^{\top}\!$. Since $E_e  \succ  0$, it follows that $W\Lambda H^{\top} E_e^{-1} H\Lambda W^{\top}  \succeq  0$. Recall that $H\Lambda H^{\top} \succ 0$ and $\Lambda \in \mathcal{D}_n$. There exists $j\in \mathbb{Z}_{[1,n]}$ such that $\lambda_j >0$. Moreover, it holds that
	\begin{equation*}
		\mathrm{tr}(\Lambda)\delta I - W\Lambda W^{\top} =
		\textstyle\sum\nolimits_{i=1}^{n}\lambda_i (\delta I - w_s(d_i)w_s(d_i)^{\!\top}).
	\end{equation*}
	{Since $w_s(d_i) \in \mathcal{W}$, we have $(\delta I - w_s(d_i)w_s(d_i)^{\!\top}) \succ 0$ for all $i \in \mathbb{Z}_{[1,n]}$. Thus, we can obtain $$\textstyle\sum\nolimits_{i=1}^{n}\lambda_i (\delta I - w_s(d_i)w_s(d_i)^{\!\top}) \!\succeq\! \lambda_j(\delta I - w_s(d_j)w_s(d_j)^{\!\top}) \!\succ\! 0,$$which implies $F_e^{\top}E_e^{-1}F_e - G_e \succ 0$.}
\end{proof}
\subsection{Proof of Lemma 11}
\begin{proof}
	For the ``if'' part, suppose that there exist matrices $Y=Y^{\top} \succ 0$, $L$ and scalar $\alpha \geq 0$ satisfying \eqref{eq:csf_criteria_1} and \eqref{eq:csf_criteria_2}. By making use of the fact that $R \succ 0$ and computing the Schur complement of \eqref{eq:csf_criteria_1} with respect to $R^{-1}$, we get $G_s \succ 0$ where
	$G_{s}:=Y-YQY-L^{\top}RL.$ Next, since $Q \succ 0$, the Schur complement of \eqref{eq:csf_criteria_2} with respect to $R^{-1}$ implies that \eqref{eq:csf_criteria_2} is equivalent to
	{\begin{equation} \label{eq:thm2_proof_1}
			\setlength{\arraycolsep}{2.3pt}
			\begin{bmatrix}
				\begin{array}{ccc:cc}
					Y & 0 & 0 & 0 & 0 \\
					0 & 0 & 0 & 0 & Y \\
					0 & 0 & 0 & 0 & L \\ \hdashline
					0 & 0 & 0 & \raisebox{0pt}[10pt]{$R^{-1}$}\! & L \\
					0 & Y^{\!\top}\! & L^{\!\top}\! & L^{\!\top}\! & Y\!-\!YQY \\
				\end{array}
			\end{bmatrix}
			\!-\alpha
			\begin{bmatrix}
				\Xi(H,\dot{X},\Lambda) & 0 \\
				0 & 0
			\end{bmatrix}
			\!\succeq \!0.
		\end{equation}
	}Then, note that the second diagonal block of \eqref{eq:thm2_proof_1} is positive definite. By computing the Schur complement of \eqref{eq:thm2_proof_1} with respect to the second diagonal block, we obtain that
	\begin{equation} \label{eq:csf_cse_criteria}
		\underbrace{
			\begin{bmatrix}
				\begin{array}{c:c}
					Y & 0 \\  \hdashline
					0  & \raisebox{0pt}[18pt]{$
						-\begin{bmatrix}
							Y \\ L
						\end{bmatrix}
						G_s^{-1}
						\begin{bmatrix}
							Y \\ L
						\end{bmatrix}^{\top}$}
				\end{array}
		\end{bmatrix}}_{:=M_s}  - \alpha \Xi(H,\dot{X},\Lambda) \succeq 0.
	\end{equation}
	By invoking Lemma \ref{lem:S-lemma}, \eqref{eq:csf_cse_criteria} implies
	\begin{equation} \label{eq:CSF_4}
		\begin{bmatrix}
			I \\
			\Delta^{\top}
		\end{bmatrix}^{\top}
		M_s
		\begin{bmatrix}
			I \\ \Delta^{\top}
		\end{bmatrix}
		\succeq 0, \ \forall \Delta \in \mathcal{C}(H,\dot{X},\Lambda),
	\end{equation}
	which is equivalent to
	\begin{equation} \label{eq:CSF_3.2}
		Y \!-\! (AY\!+\!BL)G_s^{-1}(AY\!+\!BL)^{\top} \succeq 0, \forall \Delta \!\in\! \mathcal{C}(H,\dot{X},\Lambda).
	\end{equation}
	Recall that $G_s \succ 0$. Then, by standard Schur complement argument, {\eqref{eq:CSF_3.2}} implies that
	\begin{equation} \label{eq:CSF_2}
		\begin{bmatrix}
			Y-YQY-L^{\top}RL & (AY+BL)^{\top} \\
			AY+BL & Y
		\end{bmatrix}
		\succeq 0
	\end{equation}
	holds for all $\Delta \in \mathcal{C}(H,\dot{X},\Lambda)$. By computing the Schur complement of \eqref{eq:CSF_2} with respect to the second diagonal block, we can obtain that
	\begin{equation} \label{eq:CSF_1}
		Y-(AY+BL)^{\top}Y^{-1}(AY+BL)-YQY-L^{\top}RL \succeq 0
	\end{equation}
	holds for $\forall \Delta \in \mathcal{C}$. By letting $P_T=Y^{-1}$ and $K=LP_T$,
	\begin{equation} \label{eq:convergence_stale}
		P_T - (A+BK)^{\top} P_T (A+BK) \succeq Q + K^{\top} R K
	\end{equation}
	holds for all $\Delta \in \mathcal{C}(H,\dot{X},\Lambda)$. Finally, we can conclude that \eqref{eq:lem_recursive_1} holds for all $x \in \mathbb{R}^{n_x}$ (which includes $\mathcal{X}_T$) and $ {\Delta} \in \mathcal{C}$.

	For ``only if'' part, since \eqref{eq:lem_recursive_1} holds for all $x \in \mathcal{X}_T$, we have
	\begin{equation}
		\label{eq:proof_recursive_2}
		\begin{bmatrix}
			1 \\ x
		\end{bmatrix}^{\top}
		\hspace{-2pt}
		\begin{bmatrix}
			0 & 0 \\
			0 & P_T- {A}_K^{\top}P_T  {A}_K-Q-K^{\top}RK
		\end{bmatrix}
		\begin{bmatrix}
			1 \\ x
		\end{bmatrix}
		\succeq 0
	\end{equation}
	for $\forall \Delta \in \mathcal{C}$ and $\forall x \in \mathbb{R}^{n_x}$ with
	\begin{equation}
		\label{eq:proof_recursive_3}
		\begin{bmatrix}
			1 \\ x
		\end{bmatrix}^{\top}
		\hspace{-2pt}
		\begin{bmatrix}
			L_T & 0 \\
			0 & -P_T
		\end{bmatrix}
		\begin{bmatrix}
			1 \\ x
		\end{bmatrix}
		\succeq 0
	\end{equation}
	where $A_K:=A+BK$. Note that when $x$ is a zero vector, \eqref{eq:proof_recursive_3} is positive since $L_T > 0$. Then Lemma \ref{lem:S-lemma} asserts that the claim we proposed holds if and only if there exists $\beta \geq 0$ such that
	\begin{equation}
		\label{eq:proof_recursive_4}
		\begin{bmatrix}
			0 &\hspace{-.5em} 0 \\
			0 &\hspace{-.5em} P_T\!-\! {A}_K^{\top}P_T  {A}_K\!-\!Q\!-\!K^{\!\top}\!RK
		\end{bmatrix}
		\!-\! \beta
		\begin{bmatrix}
			L_T &\hspace{-.5em} 0 \\
			0 &\hspace{-.5em} -P_T
		\end{bmatrix}
		\!\succeq\! 0.
	\end{equation}
	Recall that by $L_T > 0$, we have \eqref{eq:proof_recursive_4} holds if and only if $\beta = 0$ and $P_T- {A}_K^{\top}P_T {A}_K-Q-K^{\top}RK
	\succeq 0$ for all $\Delta \in \mathcal{C}$. This statement is equivalent to \eqref{eq:convergence_stale} and also proves \eqref{eq:csf_criteria_1} and \eqref{eq:csf_criteria_2} by letting $Y\!=\!P_{T}^{-1}$ and $L\!=\!KY$.
\end{proof}

\bibliographystyle{IEEEtran}
\bibliography{IEEEabrv,myrefs}

\end{document}